\documentclass[a4paper,USenglish,cleveref,autoref,thm-restate]{lipics-v2021}
\usepackage[textsize=footnotesize]{todonotes}
\usepackage{multirow}
\usepackage{booktabs}
\newcommand{\old}[1]{{}}
\newcommand{\eps}{\varepsilon}
\newcommand{\diam}{\mathrm{diam}}
\newcommand{\per}{\mathrm{per}}
\newcommand{\proj}{\mathrm{proj}}

\newcommand{\wdth}{\mathrm{width}}
\newcommand{\hght}{\mathrm{height}}

\newcommand{\MST}{\mathrm{MST}}
\newcommand{\dir}{\mathrm{dir}}
\newcommand{\hper}{\mathrm{hper}}
\newcommand{\vper}{\mathrm{vper}}
\newcommand{\slope}{\mathrm{slope}}
\newcommand{\shad}{\mathrm{shad}}

\usepackage{todonotes}

\hypersetup{
hidelinks,
colorlinks=true,
citecolor=[rgb]{0.121 0.47 0.705},
linkcolor=[rgb]{0.7 0 0.4},
urlcolor=[rgb]{0.121 0.47 0.705}
}
\hideLIPIcs
\nolinenumbers

\title{Euclidean Steiner Spanners: Light and Sparse}

\titlerunning{Euclidean Steiner Spanners: Light and Sparse}

\author{Sujoy Bhore}{Indian Institute of Science Education and Research, Bhopal, India.}{sujoy.bhore@gmail.com}{0000-0003-0104-1659}{}
\author{Csaba D. T\'oth}{California State University Northridge, Los Angeles, CA; and Tufts University, Medford, MA, USA.}{csaba.toth@csun.edu}{0000-0002-8769-3190}{}

\authorrunning{S.~Bhore and C.~D.~T\'oth}

\Copyright{The authors} 



\EventEditors{}
\EventNoEds{1}
\EventLongTitle{}
\EventShortTitle{}
\EventAcronym{}
\EventYear{}
\EventDate{}
\EventLocation{}
\EventLogo{}
\SeriesVolume{}
\ArticleNo{}

\begin{document}

\maketitle

\begin{abstract}
Lightness and sparsity are two natural parameters for Euclidean $(1+\eps)$-spanners. Classical results show that, when the dimension $d\in \mathbb{N}$ and $\eps>0$ are constant, every set $S$ of $n$ points in $d$-space admits an $(1+\eps)$-spanners with $O(n)$ edges and weight proportional to that of the Euclidean MST of $S$. In a recent breakthrough, Le and Solomon (2019) established the precise dependencies on $\eps>0$, for constant $d\in \mathbb{N}$, of the minimum lightness and sparsity of $(1+\eps)$-spanners, and observed that Steiner points can substantially improve the lightness and sparsity of a $(1+\eps)$-spanner. They gave upper bounds of $\tilde{O}(\eps^{-(d+1)/2})$ for the minimum lightness in dimensions $d\geq 3$, and $\tilde{O}(\eps^{-(d-1)/2})$ for the minimum sparsity in $d$-space for all $d\geq 1$.
Subsequently, Le and Solomon (2020)
constructed Steiner $(1+\eps)$-spanners of lightness $O(\eps^{-1}\log\Delta)$ in the plane, where
$\Delta\in  \Omega(\sqrt{n})$ is the \emph{spread} of $S$, defined as the ratio between the maximum and minimum distance between a pair of points.

In this work, we improve several bounds on the lightness and sparsity of Euclidean Steiner $(1+\eps)$-spanners. We establish lower bounds of $\Omega(\eps^{-d/2})$ for the lightness and $\Omega(\eps^{-(d-1)/2})$ for the sparsity of such spanners in Euclidean $d$-space for all constant $d\geq 2$. Our lower bound constructions generalize previous constructions by Le and Solomon, but the analysis substantially simplifies previous work, using new geometric insight, focusing on the directions of edges.

Next, we show that for every finite set of points in the plane and every $\eps\in (0,1]$, there exists a Euclidean Steiner $(1+\eps)$-spanner of lightness
$O(\eps^{-1})$; this matches the lower bound for $d=2$. We generalize the notion of shallow light trees, which may be of independent interest, and use directional spanners and a modified window partitioning scheme to achieve a tight weight analysis.
\end{abstract}


\section{Introduction}\label{sec:intro}
For an edge-weighted graph $G$, a subgraph $H$ of $G$ is a $t$-\emph{spanner} if $\delta_H(u,v) \le t\cdot \delta_G(u,v)$, where $\delta_G(u,v)$ denotes the shortest path distance between any two vertices $u$ and $v$.
The parameter $t$ is called the \emph{stretch factor} of the spanner.
Spanners are fundamental graph structures with many applications in the area of distributed systems and communication, distributed queuing protocol, compact routing schemes, and more; see~\cite{demmer1998arrow,
herlihy2001competitive,
PelegU89a, 
peleg1989trade}.
Two important parameters of a spanner $H$ are \emph{lightness} and \emph{sparsity}.
The \emph{lightness} of $H$ is the ratio $w(H)/w(\MST)$ between the total weight of $H$ and the weight of a minimum spanning tree (MST). The \emph{sparsity} of $H$ is the ratio $|E(H)|/|E(\MST)|\approx |E(H)|/|V(G)|$ between the number of edges of $H$ and any spanning tree. As $H$ is connected, the trivial lower bound for both the lightness and the sparsity of a spanner is 1.
When the vertices of $G$ are points in a metric space, the edge weights obey the triangle inequality. The most important examples include Euclidean $d$-space and, in general, metric spaces with constant doubling dimensions (the doubling dimension of $\mathbb{R}^d$ with $L_2$-norm is $\Theta(d)$).

In a \emph{geometric spanner}, the underlying graph $G=(S,\binom{S}{2})$ is the complete graph on a finite point set $S$ in $\mathbb{R}^d$, and the edge weights are the Euclidean distances between vertices.
Euclidean spanners are one of the fundamental geometric structures that find applications across domains, such as, topology control in wireless networks~\cite{schindelhauer2007geometric}, efficient regression in metric spaces~\cite{gottlieb2017efficient},
approximate distance oracles~\cite{gudmundsson2008approximate}, and others.
Rao and Smith~\cite{rao1998approximating} showed the relevance of Euclidean spanners in the context of other geometric \textsf{NP}-hard problems, e.g., Euclidean traveling salesman problem and Euclidean minimum Steiner tree problem, and introduced the so called \emph{banyans}\footnote{A
$(1+\eps)$-banyan for a set of points $A$ is a set of points $A'$
and line segments $S$ with
endpoints in $A\cup A'$ such that a
$1+\eps$ optimal Steiner Minimum Tree for any subset of $A$ is
contained in $S$}, which is a generalization of graph spanners.
Apart from lightness and sparsity, various other optimization criteria have been considered, e.g.,
bounded-degree spanners~\cite{bose2005constructing} and $\alpha$-diamond spanners~\cite{das1989triangulations}. Several distinct construction approaches have been developed for Euclidean spanners,
that each found further applications in geometric optimization, such as
well-separated pair decomposition (WSPD) based spanners~\cite{callahan1993optimal, GudmundssonLN02}, skip-list spanners~\cite{arya1994randomized},
path-greedy and gap-greedy spanners~\cite{althofer1993sparse, arya1997efficient}, and more. For an excellent survey of  results and techniques on Euclidean spanners up to 2007, we refer to the book by Narasimhan and Smid~\cite{narasimhan2007geometric}.

\subparagraph{Sparsity.}
A large body of research on spanners has been devoted to \emph{sparse spanners} where the objective is to obtain a spanner with small number of edges, preferably $O(|S|)$, with $1+\eps$ stretch factor, for any given $\eps>0$.
Chew~\cite{Chew86} was the first to show that there exists a Euclidean spanner with a linear number of edges and stretch factor $\sqrt{10}$. The stretch factor was later improved to $2$~\cite{Chew89}. 
Later, Keil and Gutwin~\cite{keil1992classes} showed that the Delanauy triangulation of the point set $S$ is a $2.42$-spanner. 
Clarkson~\cite{Clarkson87} designed the first Euclidean $(1+\eps)$-spanner, for arbitrary small $\eps>0$; an alternative algorithm was  presented by Keil~\cite{keil1988approximating}.
Moreover, these papers introduced the fixed-angle $\Theta$-graph\footnote{The $\Theta$-graph is a type of geometric spanner similar to Yao graph~\cite{yao1982constructing}, where the space around each point $p\in P$ is partitioned into cones of angle $\Theta$, and $S$ will be connected to a point $q\in P$ whose orthogonal projection to some fixed ray contained in the cone is closest to $S$.} as a potential new tool for designing spanners in $\mathbb{R}^2$, which was later generalized to higher dimension by Ruppert and Seidel~\cite{ruppert1991approximating}. One can construct a $(1+\eps)$-spanner with $O(n\eps^{-d+1})$ edges by taking the angle $\Theta$ to be proportional to $\eps$ in any constant dimension $d\geq 1$.
Recently, Le and Solomon~\cite{le2019truly} showed that this bound is tight, as for every $\eps>0$ and constant $d\in \mathbb{N}$, there are sets of $n$ points in $\mathbb{R}^d$ for which any $(1+\eps)$-spanner must have sparsity $\Omega(\eps^{-d+1})$, whenever $\eps = \Omega(n^{-1/(d-1)})$.

\subparagraph{Lightness.}
For a set of points $S$ in a metric space, the  lightness is the ratio of the spanner weight (i.e., the sum of all edge weights) to the weight of the minimum spanning tree $\MST(S)$.
Das et al.~\cite{das1993optimally} showed that the \emph{greedy-spanner}, introduced by Alth\"ofer et al.~\cite{althofer1993sparse}, has constant lightness in $\mathbb{R}^3$ for any constant $\eps>0$. This was generalized later to $\mathbb{R}^d$, for all $d\in \mathbb{N}$, by Das et al.~\cite{narasimhan1995new}. However the dependency on the parameter $\eps$ (for constant $d$) has not been addressed. Rao and Smith showed that the greedy-spanner has lightness $\eps^{-O(d)}$ in $\mathbb{R}^d$ for every constant $d$, and asked what is the best possible constant in the exponent.
A detailed analysis in the book on geometric spanners~\cite{narasimhan2007geometric} shows that the lightness of the greedy-spanner is $O(\eps^{-2d})$ in $\mathbb{R}^d$.
%
Le and Solomon~\cite{le2019truly} showed that the greedy-spanner has lightness $O(\eps^{-d}\log \eps^{-1})$  in $\mathbb{R}^d$.
Moreover, they constructed, for every $\eps>0$ and constant
$d\in \mathbb{N}$, a set $S$ of $n$ points in $\mathbb{R}^d$ for which any $(1+\eps)$-spanner must have lightness $\Omega(\eps^{-d})$, whenever $\eps = \Omega(n^{-1/(d-1)})$.
Recently, Borradaile et al.~\cite{borradaile2019greedy} showed that the greedy-spanner of a finite metric space of doubling dimension $d$ has lightness $\eps^{-O(d)}$.

\subparagraph{Euclidean Steiner spanners.}
Steiner points are additional vertices in a network (via points) that are not part of the input, and a $t$-spanner must achieve stretch factor $t$ only between pairs of the input points in $S$.
A classical problem on Steiner points arises in the context of minimum spanning trees. The \emph{Steiner ratio} is the supremum ratio between the weight of a \emph{minimum Steiner tree} and a \emph{minimum spanning tree} of a finite point set, and it is at least $\frac{1}{2}$ in any metric space due to the triangle inequality.

Le and Solomon~\cite{le2019truly} noticed that Steiner points can substantially improve the bound on the lightness and sparsity of an $(1+\eps)$-spanner. Previously, Elkin and Solomon~\cite{elkin2015steiner} and Solomon~\cite{Solomon15} showed that Steiner points can improve the weight of the network in the single-source setting. In particular, the so-called \emph{shallow-light trees} (\textsf{SLT}) is a single-source spanning tree that concurrently approximates a shortest-path tree (between the source and all other points) and a minimum spanning tree (for the total weight). They proved that Steiner points help to obtain exponential improvement on the lightness of \textsf{SLT}s in a general metric space~\cite{elkin2015steiner}, and quadratic improvement on the lightness in Euclidean spaces~\cite{Solomon15}.

Le and Solomon, used Steiner points to improve the bounds for lightness and sparsity of Euclidean spanners. For minimum sparsity, they gave an upper bound of $O(\eps^{(1-d)/2})$ for $d$-space and a lower bound of $\Omega(\eps^{-1/2}/\log\eps^{-1})$ in the plane ($d=2$)~\cite{le2019truly}.
For minimum lightness, Le and Solomon~\cite{le2020light} gave an upper bound of $O(\eps^{-1}\log\Delta)$ in the plane and $O(\eps^{-(d+1)/2}\log\Delta)$ in dimension $d\geq 3$, where $\Delta$ is the \emph{spread} of the point set, defined as the ratio between the maximum and minimum distance between a pair of points. In any space with doubling dimension $d$ (including $\mathbb{R}^d$), we have $\log \Delta\geq \Omega_d(\log n)$. 
Subsequently, Le and Solomon~\cite{le2020light-arxiv} noted that the factor $\log \Delta$ can be improved to $\log n$ using standard techniques.
Moreover, Le and Solomon~\cite{le2020unified} constructed Steiner $(1+\eps)$-spanners with lightness $\tilde{O}(\eps^{-(d+1)/2})$ in dimensions $d\geq 3$.
Recently, we have studied online spanners in both Euclidean and general metrics, and obtained several asymptotically tight bounds~\cite{BT-ons-21, BFKT22}.


\begin{table}
\begin{center}
\begin{tabular}{ |c|c|c| }
\hline
Bounds & Sparsity & Lightness \\
\hline
\multirow{2}{5em}{Lower Bounds}
& $\Omega(\varepsilon^{-1/2}/\log \varepsilon^{-1})$ for $d=2$~\cite{le2019truly} & $\Omega(\varepsilon^{-1}/\log \varepsilon^{-1})$ for $d=2$~\cite{le2019truly}\\
& \textcolor{magenta}{$\Omega(\varepsilon^{(1-d)/2})$}~\cite{BT-less-21} &
\textcolor{magenta}{$\Omega(\varepsilon^{-d/2})$ for $d\geq 2$}~\cite{BT-less-21}\\
\hline
\multirow{3}{5em}{Upper Bounds} &  & $\tilde{O}(\varepsilon^{-(d+1)/2})$ for $d\ge 3$~\cite{le2020unified}\\
&  $O(\varepsilon^{(1-d)/2})$~\cite{le2019truly} & $O(\varepsilon^{-1}\log n)$ for $d=2$~\cite{le2020light}\\
&  & \textcolor{magenta}{$O(\varepsilon^{-1})$ for $d=2$}~\cite{BT-lessp-20}\\
\hline
\end{tabular}
\vspace{.2cm}
\caption{\label{tab:table-name} Previous and new results on Euclidean Steiner $(1+\eps)$-spanners; new results are highlighted in magenta. The $\Omega(.)$ and $O(.)$ notation hides constant coefficients dependent on $d$; and $\tilde{O}(.)$ also hides polylogarithmic factors in $\eps$.}
\end{center}
\end{table}

\subparagraph{Our Contribution.}
We improve the bounds on the lightness and sparsity of Euclidean Steiner $(1+\eps)$-spanners; see Table~\ref{tab:table-name}. First, in Section~\ref{sec:lower}, we prove the following lower bounds.

\begin{restatable}{theorem}{lowerboundth}
\label{thm:lb}
Let $d\in \mathbb{N}$, $d\geq 2$, be a constant and let $\varepsilon>0$.
For every integer $n\geq \Omega(\eps^{(1-d)/2})$, there exists a set $S$ of $n$ points in $\mathbb{R}^d$ such that any Euclidean Steiner $(1+\eps)$-spanner for $S$ has
lightness $\Omega(\eps^{-d/2})$ and sparsity $\Omega(\eps^{(1-d)/2})$.	
\end{restatable}

For lightness in dimension $d=2$, this improves the earlier bound of $\Omega(\eps^{-1}/\log \eps^{-1})$ by Le and Solomon~\cite{le2019truly} by a logarithmic factor; and it is the first lower bound in dimensions $d\geq 3$. The point set $S$ in Theorem~\ref{thm:lb} is fairly simple: It consists of two square grids in two parallel hyperplanes in $\mathbb{R}^d$. However, our lower-bound analysis is significantly simpler than that of~\cite{le2019truly}. In particular, our analysis does not depend on planarity, and it generalizes to higher dimensions. The key new insight pertains to a geometric property of Steiner $(1+\eps)$-spanners: If the length of an $ab$-path $S$ between points $a,b\in\mathbb{R}^d$ is at most $(1+\eps)\|ab\|$, then ``most'' of the edges of $S$ are almost parallel to $ab$. We expand on this idea in Section~\ref{sec:pre}.

Then, in Section~\ref{sec:redux} we prove the following theorem on light spanners.

\begin{restatable}{theorem}{upperboundtheorem}
\label{thm:upper}
For every set $S$ of $n$ points in Euclidean plane and every $\eps\in (0,1)$,
there exists a Steiner $(1+\eps)$-spanner of
lightness $O(\eps^{-1})$.
\end{restatable}

This result improves on an earlier bound of $O(\eps^{-1}\log \Delta)$ by Le and Solomon~\cite{le2020light}, where $\Delta$ is the \emph{spread} of the point set, defined as the ratio between the maximum and minimum distance between a pair of points. Note that $\Delta\geq n^{\Omega(1/d)}$ in a metric space of doubling dimension $d$.

The tight bound in Theorem~\ref{thm:upper} relies on three new ideas, which may be of independent interest:
First, we generalize Solomon's \textsf{SLT}s to points on a staircase path or a monotone rectilinear path (Section~\ref{sec:SLT}).
Second,  we reduce the proof of Theorem~\ref{thm:upper} to the construction of ``directional'' spanners, in each of $\Theta(\eps^{-1/2})$ directions, where it is enough to establish the stretch factor $1+\eps$ for point pairs $s,t\in S$ where the direction of $st$ is in an interval of size $\sqrt{\eps}$ (Section~\ref{sec:redux}).
Combining the first two ideas, we show how to construct light directional spanners for points on a staircase path (Section~\ref{sec:staircases}).
In each direction, we start with a rectilinear MST of $S$, and augment it into a directional spanner.
We refine the classical window partition of a rectilinear polygon into histograms by subdividing each histogram into special histograms (called tame histograms), whose boundary does not oscillate wildly; this is the final piece of the puzzle. These histograms are sufficiently flexible to keep the total weight of the subdivision under control, and we can construct directional $(1+\eps)$-spanners for each face of such a subdivision (Sections~\ref{sec:tame}--\ref{sec:thin}).

\section{Preliminaries}
\label{sec:pre}

Let $d\geq 2$ be an integer, and $S$ a set of $n$ points in $\mathbb{R}^d$. For $a,b\in \mathbb{R}^d$, the Euclidean distance between $a$ and $b$ is denoted by $\|ab\|$. For a set $E$ of line segments in $\mathbb{R}^d$, let $\|E\|=\sum_{e\in E}\|e\|$ be the total weight of all segments in $E$. For a geometric graph $G=(S,E)$, where $S\subset \mathbb{R}^d$, we also use the notation $\|G\|=\|E\|$, which is the Euclidean weight of graph $G$.

We briefly review a few geometric primitives in $d$-space.
For $a,b\in \mathbb{R}^d$, the locus of points $c\in \mathbb{R}^d$ with $\|ac\|+\|cb\|\leq (1+\eps)\|ab\|$ is an ellipsoid $\mathcal{E}_{ab}$ with foci $a$ and $b$, and major axis of length $(1+\eps)\|ab\|$; see Fig.~\ref{fig:ellipse}(a).
Note that all $d-1$ minor axes of $\mathcal{E}_{ab}$ are $\sqrt{(1+\eps)^2-1}\|ab\|=\sqrt{2\eps+\eps^2}\|ab\|<\sqrt{3\eps}\|ab\|$ when $\eps<1$.
In particular, the aspect ratio of the minimum bounding box of $\mathcal{E}_{ab}$ is roughly  $\sqrt{\eps}$. By the triangle inequality, $\mathcal{E}_{ab}$ contains every  $ab$-path of weight at most $(1+\eps)\|ab\|$.

\begin{figure}[htbp]
 \centering
 \includegraphics[width=\textwidth]{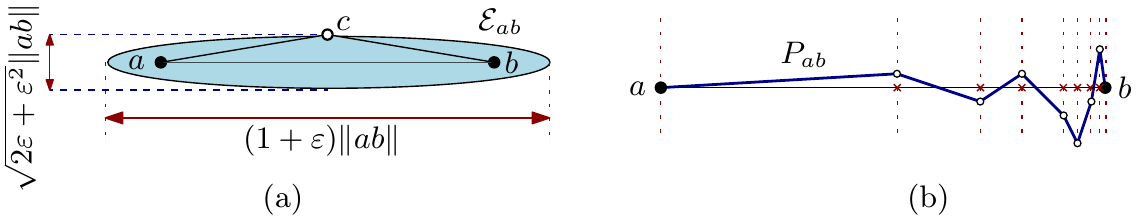}
 \caption{(a) An ellipse $\mathcal{E}_{ab}$ with foci $a$ and $b$, and major axis $(1+\eps)\|ab\|$.
 (b) A monotone $ab$-path $P_{ab}$, and the projections of its edges to $ab$.}
    \label{fig:ellipse}
\end{figure}

The unit vectors in $\mathbb{R}^d$ are on the $(d-1)$-sphere $\mathbb{S}^{d-1}$; the direction vectors of a line in $\mathbb{R}^d$ can be represented by vectors of a hemisphere. The \emph{angle} between two unit vectors, $\overrightarrow{u}_1$ and $\overrightarrow{u}_2$ is
$\angle(\overrightarrow{u}_1,\overrightarrow{u}_2)=\arccos(\overrightarrow{u}_1\cdot \overrightarrow{u}_2)\in (-\pi,\pi)$. Between two (undirected) edges $e_1$ and $e_2$ with unit direction vectors $\pm\overrightarrow{u}_1$ and $\pm\overrightarrow{u}_2$, we define the angle as $\angle (e_1,e_2)=\arccos | \overrightarrow{u}_1\cdot \overrightarrow{u}_2|\in [0,\pi)$. A path $(v_0,v_1,\ldots , v_m)$ in $\mathbb{R}^d$ is \emph{monotone in direction $\overrightarrow{u}$} if $\overrightarrow{v_{i-1}v_i}\cdot \overrightarrow{u}\geq 0$ for all $i\in \{1,\ldots, m\}$; and it is simply \emph{monotone} if it is monotone in direction $\overrightarrow{v_0v_m}$. Let $\proj_{ab}(e)$ denote the orthogonal projection of an edge $e$ to the supporting line of $ab$, see Fig.~\ref{fig:ellipse}(b); and note that $\|\proj_{ab}(e)\|=\|e\|\cos\angle(ab,e)$.

In Euclidean plane ($d=2$), we can parameterize the directions by angles.
The \emph{direction} of a line segment $ab$ in $\mathbb{R}^2$, denoted $\mathrm{dir}(ab)$, is the minimum counterclockwise angle $\alpha\in [0,\pi)$ that rotates the $x$-axis to be parallel to $ab$. The set of possible directions $[0,\pi)$ is homeomorphic to the unit circle $\mathbb{S}^1$, and an interval $(\alpha,\beta)$ of directions corresponds to the counterclockwise arc of $\mathbb{S}^1$ from $\alpha\pmod{\pi}$ to $\beta\pmod{\pi}$.
A path in $\mathbb{R}^2$ is \emph{$x$-monotone} (resp., \emph{$y$-monotone}) if it is monotone in direction $\overrightarrow{u}=(1,0)$ (resp., $\overrightarrow{u}=(0,1)$). A \emph{staircase path} is simple path that is both $x$- and $y$-monotone.
The \emph{width} and \emph{height} of a path or a polygon $P$ is the Euclidean length of its orthogonal projection to the $x$-axis and $y$-axis, respectively.

\subparagraph{Angle-Bounded Paths.}
For $\delta\in (0,\pi/2]$, a polygonal path $(v_0,\ldots, v_m)$ in $\mathbb{R}^d$ is \emph{$(\theta\pm \delta)$-angle-bounded} if the direction of every segment $v_{i-1}v_i$ is in the interval $[\theta-\delta,\theta+\delta]$.
Borradaile and Eppstein~\cite[Lemma~5]{BorradaileE15} observed that the
weight of a $(\theta\pm \delta)$-angle-bounded $st$-path is at most $(1+O(\delta^2))\|st\|$.
We prove this observation in a more precise form in $\mathbb{R}^d$.
The quadratic growth rate in $\delta$ is due to the Taylor estimate $\sec(x)=\frac{1}{\cos(x)}\leq 1+x^2$ for $0\leq x\leq \frac{\pi}{4}$.

\begin{lemma}\label{lem:angle2}
Let $a,b\in \mathbb{R}^d$ and let $P=(v_0,v_1, \ldots ,v_m)$ be an $ab$-path such that $P$ is monotonic in direction $\overrightarrow{ab}$ and $\angle(ab,v_{i-1}v_i)\leq \delta\leq \frac{\pi}{4}$, for $i=1,\ldots ,m$. Then $\|P\|\leq (1+\delta^2)\|ab\|$.
\end{lemma}
\begin{proof}
For $i=0,\ldots , m$, let $u_i$ be the orthogonal projection of $v_i$ to the line $ab$, and let $\alpha_i=\angle(ab,v_{i-1}v_i)$. Since $P$ is monotonic in direction $\overrightarrow{ab}$, then
$\|P\| =\sum_{i=1}^m \|v_{i-1}v_i\|
=\sum_{i=1}^m \|u_{i-1}u_i\|\sec \angle(ab,v_{i-1}v_i)
\leq \sec \delta\cdot \sum_{i=1}^m \|u_{i-1}u_i\|
\leq  (1+\delta^2)\|ab\|$, as claimed.
%
\end{proof}

\subparagraph{Characterization for Short $ab$-Paths.}
Let $a,b\in \mathbb{R}^d$, and let $P_{ab}$ be a polygonal $ab$-path of weight at most $(1+\eps)\|ab\|$. We show that ``most'' edges along $P_{ab}$ must be ``nearly'' parallel to $ab$. Specifically, for an angle $\alpha\in [0,\pi/2)$, we distinguish between two types of edges in $P_{ab}$. Denote by $E(\alpha)$ the set of edges $e$ in $P_{ab}$ with $\angle(ab,e)<\alpha$; and let $F(\alpha)$ be the set of all other edges of $P_{ab}$.
Clearly, we have $\|P_{ab}\|=\|E(\alpha)\|+\|F(\alpha)\|$ for all $\alpha$.

\begin{lemma}\label{lem:parallel}
Let $a,b\in \mathbb{R}^d$ and let $P_{ab}$ be an $ab$-path of weight $\|P_{ab}\|\leq (1+\eps)\|ab\|$. Then for every $i\in \{1,\ldots, \lfloor \frac{\pi}{2}/ \sqrt{\eps}\rfloor\}$,
we have $\|E(i\cdot \sqrt{\eps})\|\geq (1-2/i^2)\,\|ab\|$.
\end{lemma}
\begin{proof}
Suppose, to the contrary, that $\|E(i\cdot \sqrt{\eps})\|<(1-2/i^2)\,\|ab\|$ for some $i\in \{1,\ldots, \lfloor\frac{\pi}{2}/\sqrt{\eps}\rfloor\}$. We have
\begin{equation}\label{eq:proj}
\sum_{e\in E(i\,\sqrt{\eps})\cup F(i\,\sqrt{\eps})}  \| \text{proj}_{ab}(e)\|\geq \|ab\|,
\end{equation}
which implies
\begin{align}\label{eq:proj2}
\sum_{e\in F(i\,\sqrt{\eps})} \| \text{proj}_{ab}(e)\|
    &\geq  \|ab\|-\sum_{e\in E(i\,\sqrt{\eps})} \| \text{proj}_{ab}(e)\| \\
    &\geq  \|ab\|-\sum_{e\in E(i\,\sqrt{\eps})} \| e\| \nonumber\\
    &= \|ab\|-\|E(i\,\sqrt{\eps})\|. \nonumber
\end{align}
Recall that for every edge $e\in F(i\,\sqrt{\eps})$, we have $\angle(e,ab)\geq i\cdot \sqrt{\eps}$. Using the Taylor estimate $\frac{1}{\cos(x)}\geq 1+\frac{x^2}{2}$, for every $e\in F(i\,\sqrt{\eps})$, we obtain
\[|e\|\geq \frac{\|\text{proj}_{ab}(e)\|}{\cos(i\cdot \sqrt{\eps})}
\geq \|\text{proj}_{ab}(e)\|\left(1+\frac{(i\,\sqrt{\eps})^2}{2}\right)
= \|\text{proj}_{ab}(e)\|\left(1+\frac{i^2\,\eps}{2}\right),\]
Combined with \eqref{eq:proj2}, this yields
\begin{align*}
\|P_{ab}\|
    &= \sum_{e\in E(i\,\sqrt{\eps})} \|e\| +  \sum_{e\in F(i\,\sqrt{\eps})} \|e\| \\
    &\geq \sum_{e\in E(i\,\sqrt{\eps})} \|e\| +  \sum_{e\in F(i\,\sqrt{\eps})} \|\text{proj}_{ab}(e)\|\left(1+\frac{i^2\,\eps}{2}\right) \\
    &\geq \|E(i\,\sqrt{\eps})\| +\big(\|ab\|-\|E(i\,\sqrt{\eps})\|\big) \left(1+\frac{i^2\,\eps}{2}\right)\\
    &= \left(1+\frac{i^2\,\eps}{2}\right) \|ab\|
        - \frac{i^2\,\eps}{2}\, \|E(i\,\sqrt{\eps})\|\\
    &> \left(1+\frac{i^2\,\eps}{2}\right)  \|ab\|
        -\frac{i^2\,\eps}{2} \left(1-\frac{2}{i^2}\right)\,\|ab\|\\
    &\geq \left(1+\frac{i^2\,\eps}{2}\right) \|ab\| -\left(\frac{i^2}{2}-1\right)\eps\,\|ab\| \\
    &= (1+\eps) \|ab\|,
\end{align*}
which is a contradiction, and completes the proof.
\end{proof}

We use Lemma~\ref{lem:parallel} in the analysis of our lower bound construction in Section~\ref{sec:lower}.
We can also derive a converse of Lemma~\ref{lem:parallel} for monotone $ab$-paths. An $ab$-path is \emph{monotone} if $\angle(\overrightarrow{ab},\overrightarrow{e})>0$ for every directed edge $\overrightarrow{e}$ of $P_{ab}$, where the path is directed from $a$ to $b$. Equivalently, an $ab$-path is monotone if it crosses every hyperplane orthogonal to $ab$ at most once.
We show that if the angle $\angle(\overrightarrow{ab},\overrightarrow{e})$ is sufficiently small for ``most'' of the directed edges $\overrightarrow{e}$ of $P_{ab}$, then $\|P_{ab}\|\leq (1+\eps)\|ab\|$.

\begin{lemma}\label{lem:parallel+}
For every $\delta>0$, there is a $\kappa>0$ with the following property.
For $a,b\in \mathbb{R}^d$ and a monotone $ab$-path $P_{ab}$,
if $\|F(i\cdot\sqrt{\eps\kappa})\|\leq \|P_{ab}\|/i^{2+\delta}$ for all $i\in\{1,\ldots ,\lceil \frac{\pi}{2}/\sqrt{\eps\kappa}\rceil\}$, then $\|P_{ab}\|\leq (1+\eps)\|ab\|$.
\end{lemma}
\begin{proof}
Let $P_{ab}$ be an $ab$-path with edge set $E$. Note that, by definition, $F(0)=E$.
For angles $0\leq \alpha<\beta\leq \pi/2$, let $E(\alpha,\beta)$ denote the set of edges $e\in E$ with $\alpha\leq \angle (ab,e)<\beta$. For convenience, we put $m=\lceil \frac{\pi}{2}/\sqrt{\eps\kappa}\rceil$. Using the Taylor estimate $\cos x\geq 1-x^2/2$, we can bound the excess weight of $P_{ab}$ as follows.
\begin{align*}
\|P_{ab}\|-\|ab\|
&= \sum_{e\in E} \|e\| - \sum_{e\in E} \|\proj_{ab}e\|\\
&= \sum_{e\in E} \|e\| (1-\cos \angle (ab,e))\\
&\leq \sum_{i=1}^m \|E((i-1)\sqrt{\eps\kappa},i\,\sqrt{\eps\kappa})\| (1-\cos (i\cdot\sqrt{\eps\kappa})) \\
&\leq \sum_{i=1}^m \|E((i-1)\sqrt{\eps\kappa},i\,\sqrt{\eps\kappa})\| \cdot \frac{i^2\,\eps\kappa}{2} \\
&\leq \sum_{i=1}^m \left(\|F((i-1)\sqrt{\eps})\| - \|F(i\,\sqrt{\eps})\|\right)
    \cdot \frac{i^2\,\eps\kappa}{2} \\
&= F(0)\cdot \frac{1^2\eps\kappa}{2}+\sum_{i=1}^m \|F(i\,\sqrt{\eps\kappa})\|
    \left(\frac{(i+1)^2\,\eps\kappa}{2}-  \frac{i^2\,\eps\kappa}{2}\right) \\
&\leq \|P_{ab}\|\cdot \frac{\eps\kappa}{2}+\sum_{i=1}^m \frac{\|P_{ab}\|}{i^{2+\delta}}\cdot (2i+1)\eps\kappa \\
&\leq \eps\kappa\cdot \|P_{ab}\| \left(\frac12+\sum_{i=1}^\infty \frac{2i+1}{i^{2+\delta}}\right).
\end{align*}
For $\kappa= \frac12(\frac12+\sum_{i=1}^\infty (2i+1)/2^{2+\delta})^{-1}$, we obtain
\[\|P_{ab}\|-\|ab\|\leq \frac{\eps}{2}\,\|P_{ab}\|,\]
which readily implies $\|P_{ab}\|\leq (1-\eps/2)^{-1}\|ab\|<(1+\eps)\|ab\|$,
as required.
\end{proof}

The criteria in Lemma~\ref{lem:parallel+} can certify that a geometric graph $G$ is a Euclidean Steiner $(1+\eps)$-spanner for a point set $S$. Intuitively, a geometric graph is a Steiner $(1+\eps)$-spanner for $S$ it it contains, for all point pairs $a,b\in S$, a monotone $ab$-path in which the majority of edges $e$ satisfy $\angle(ab,e)\leq O(\sqrt{\eps})$, with exceptions quantified by Lemma~\ref{lem:parallel+}. This property has already been used implicitly by Solomon~\cite{Solomon15} in the single-source setting, for the design of shallow-light trees.
We use shallow-light trees in our upper bound (Section~\ref{sec:redux}), instead of Lemma~\ref{lem:parallel+}. However, the characterization of $ab$-paths of weight at most $(1+\eps)\|ab\|$, presented in this section, may be of independent interest.

\old{
\begin{figure}[htbp]
 \centering
 \includegraphics[width=0.98\textwidth]{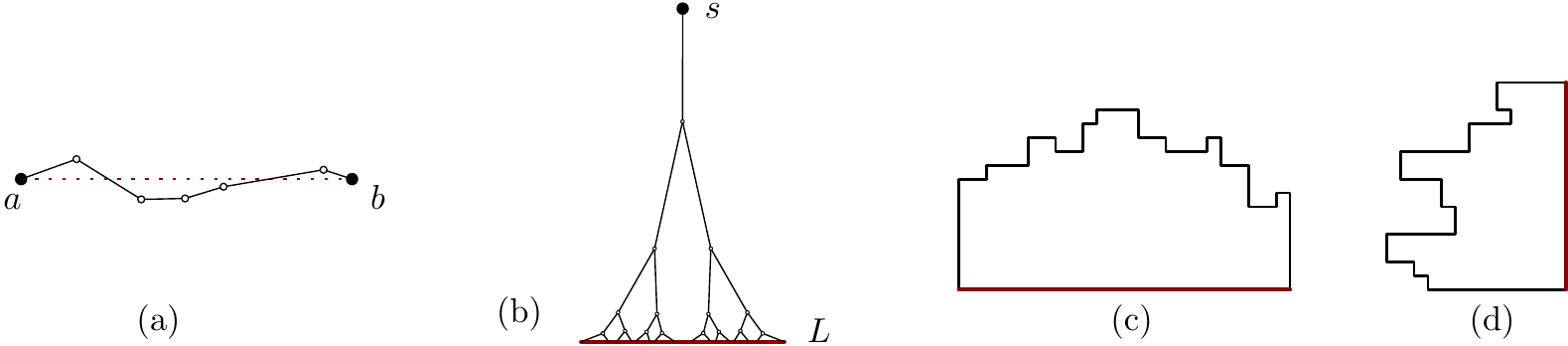}
 \caption{(a) A $(0\pm\delta)$-angle-bounded path.
(b) A shallow-light tree between a source $s$ and a horizontal line segment $L$.
(c)--(d) An $x$- and a $y$-monotone histogram.}
    \label{fig:1}
\end{figure}
}

\section{Lower Bounds}
\label{sec:lower}

In this section we prove the following result.

\lowerboundth*


\begin{proof}
First we establish the result for a point set of size $\Theta_d(\eps^{(1-d)/2})$ and then generalize to arbitrary $n$. We may assume that $0<\eps<(8d)^{-2}$.
Let $Q=[0,1]^d$ be a unit cube in $\mathbb{R}^d$; see Fig.~\ref{fig:twogrids}.
The point set $S$ will consist of two square grids in two opposite faces of $Q$, with roughly $8d\cdot\sqrt{\eps}$ spacing. Specifically, let $\varrho=\lceil \frac{1}{8d\cdot\sqrt{\eps}}\rceil$ and  consider the lattice $L=\varrho^{-1} \cdot \mathbb{Z}^d$. Let $Q_0$ and $Q_1$, respectively, be the two faces of $Q$ orthogonal to the $x_d$-axis. Now let $S_0=L\cap Q_0$ and $S_1=L\cap Q_1$. We have $|S_0|=|S_1|=(\varrho+1)^{d-1}=\Theta_d(\eps^{(1-d)/2})$,
hence $|S|=\Theta_d(\eps^{(1-d)/2})$.

\begin{figure}[htbp]
 \centering
 \includegraphics[width=0.44\textwidth]{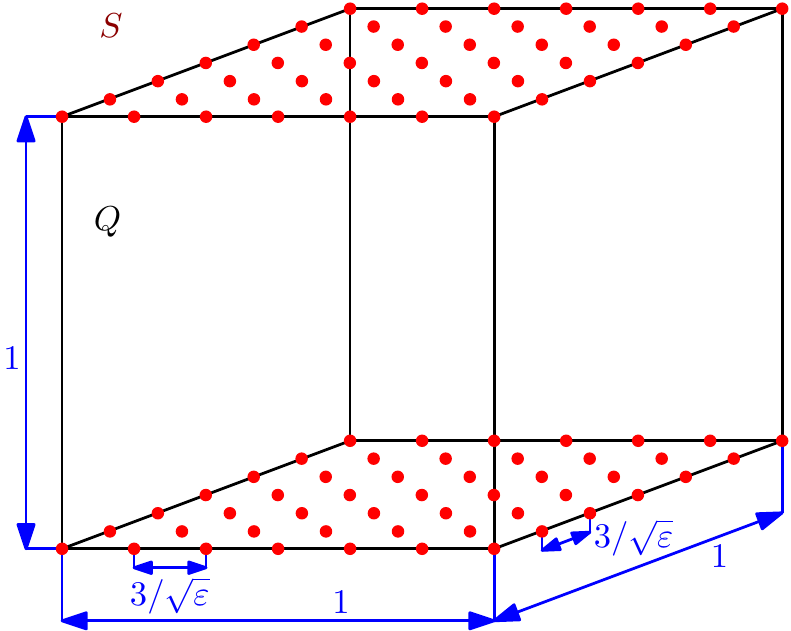}
 \caption{A schematic image of $S$ in $\mathbb{R}^3$.}
    \label{fig:twogrids}
\end{figure}

Let $N$ be a Euclidean Steiner $(1+\eps)$-spanner for $S$. For any pair of points $(a,b)\in S_0\times S_1$, the spanner $N$ contains an $ab$-path $P_{ab}$ of weight at most $(1+\eps)\|ab\|$, which lies in the ellipsoid $\mathcal{E}_{ab}$ with foci $a$ and $b$, and major axis $(1+\eps)\|ab\|$. The ellipsoid $\mathcal{E}_{ab}$ is, in turn, contained in an infinite cylinder $C_{ab}$ with axis $ab$ and radius $\frac{1}{2}\cdot\sqrt{(1+\eps)^2-1^2}\cdot \|ab\| < \sqrt{\eps}\cdot\diam(Q) = \sqrt{d\eps}$. The intersection of the cylinder $C_{ab}$ with hyperplanes containing $Q_0$ and $Q_1$, resp., is an ellipsoid of half-diameter less than $\sqrt{d\eps}/\cos\angle(ab,x_d) \leq \sqrt{d\eps}\cdot \sqrt{d}\leq d\cdot \sqrt{\eps}$, and their centers are  $a$ and $b$, respectively. Hence, all point in $S$, other than $a$ and $b$, are in the exterior of $C_{ab}$.

We distinguish between two types of edges in the $ab$-path $P_{ab}$.
An edge $e$ of $P_{ab}$ is \emph{near-parallel} to $ab$ if $\angle(ab,e)<2\cdot\sqrt{\eps}$.
Let $E(ab)$ be the set of edges of $P_{ab}$ that are near-parallel to $ab$, and $F(ab)$ the set of all other edges of $P_{ab}$. Lemma~\ref{lem:parallel} with $i=2$ yields
\begin{equation}\label{eq:12}
    \|E(ab)\|\geq \frac12 \|ab\|\geq \frac12.
\end{equation}

Notice that for two pairs $(a_1,b_1), (a_2,b_2)\in S_0\times S_1$, if $\{a_1,b_1\}\neq \{a_2,b_2\}$, then $E(a_1b_1)\cap E(a_2b_2)=\emptyset$. Indeed, in case $\angle(a_1 b_1,a_2 b_2) \geq 4\sqrt{\eps}$, this follows from the fact that the \emph{directions} near-parallel to $a_1b_1$ and $a_2b_2$, resp., are disjoint. Assume now that $\angle(a_ 1b_1,a_2 b_2) < 4\sqrt{\eps}$. Translate $a_2b_2$ to a line segment $a_1c_1$. Then we have $c_1\in L$, and the sine theorem in the triangle $\Delta(a_1b_1c_1)$ yields
\[
\|a_1c_1\|
=\|a_1b_1\| \frac{\sin\angle(b_1a_1, b_1c_1)}{\sin\angle (c_1a_1,c_1b_1)}
\leq \diam(Q)\frac{\sin \angle(a_1 b_1,a_2 b_2)}{1}
\leq \sqrt{d}\cdot \sin\left(4\sqrt{\eps}\right)
<4\cdot \sqrt{\eps d}.
\]
However, the minimum distance between any two points in the lattice $L$ is
$\varrho^{-1} = \lceil \frac{1}{8d\cdot \sqrt{\eps}}\rceil^{-1}$.
Since  $\varrho=\lceil \frac{1}{8d\cdot \sqrt{\eps}}\rceil< \frac{2}{8d\cdot \sqrt{\eps}}\leq \frac{1}{4\cdot \sqrt{\eps d}}$ for $0<\eps<(8d)^{-2}$,  then $b_1$ and $c_1$ cannot be distinct lattice points.
Therefore $c_1=b_1$, hence $a_2b_2$ is parallel to $a_1b_1$. Consequently the cylinders $C_{a_1 b_1}$ and $C_{a_2 b_2}$ have disjoint interiors, and so $E(a_1b_1)\cap E(a_2b_2)=\emptyset$, as claimed.
Combined with \eqref{eq:12}, this yields
\begin{equation}\label{eq:weight}
\|N\|
\geq \sum_{(a,b)\in S_0\times S_1} \|E(ab)\|
\geq |S_0|\cdot |S_1|\cdot \frac12
\geq \Theta_d( \eps^{1-d}).
\end{equation}

Similarly to~\cite[Claim~5.3]{le2019truly}, we may assume that $N\subseteq Q$ (indeed, we can replace every vertex of $N$ outside of $Q$ by the closest point in the boundary of $Q$; such replacements do not increase the weight of $N$).
In follows that the weight of every edge is at most $\text{diam}(Q)=\sqrt{d}$.
Consequently,
\[|E(N)|\geq \frac{\|N\|}{\max_{e\in E(N)} \|e\|} =\frac{\Omega_d(\eps^{1-d})}{\sqrt{d}}=\Omega_d(\eps^{1-d}).\]
The sparsity of $N$ is
$|E(N)|/|S|=\Omega_d(\eps^{1-d}/\eps^{(1-d)/2})=\Omega_d(\eps^{(1-d)/2})$, as required.

The MST for the point set $S$ contains one unit-weight edge between $S_0$ and $S_1$, and the remaining $|S|-2$ edges each have weight $d\sqrt{\eps}$, which is the minimum distance between lattice points in $L$ (see~\cite{SteeleS89} for the asymptotic behavior of the MST of a section of the lattice). Therefore $\|\MST(S)\|=1+(|S|-2)d\sqrt{\eps}=\Theta_d(\eps^{1-d/2})$.
It follows that the lightness of $N$ is $\|N\|/\|\MST(S)\|=\Omega_d(\eps^{1-d}/\eps^{1-d/2})=\Omega_d(\eps^{-d/2})$, as claimed.
This completes the proof when $n=\Theta_d(\eps^{(1-d)/2})$.

\subparagraph{General Case.}
Let $S_0$ denote the above construction with $|S_0|=\Theta_d(\eps^{(1-d)/2})$. Finally, if $n> |S_0|$, we can generalize the construction by duplication. Assume w.l.o.g.\ that $n=k\, |S_0|$ for some integer $k\geq 1$.
Let $Q_1,\ldots ,Q_k$, be disjoint axis-aligned unit hypercubes,
such that they each have an edge along the $x_1$-axis,
and two consecutive cubes are at distance 3 apart.
Let $S$ be the union of $k$ translates of the point set $S_0$,
on the boundaries of $Q_1,\ldots, Q_k$.
Let $N$ be a Euclidean Steiner $(1+\eps)$-spanner for $S$;
and $N_i=N\cap Q_i$ for $i=1,\ldots k$.

Since the ellipsoids induced by point pairs in different copies of $S_0$ are disjoint,
we have $\|N\|\geq \sum_{i=1}^k \|N_i\|=\Omega_d(k\eps^{1-d})$ and $|E(N)|\geq \sum_{i=1}^k |E(N_i)|$.
This immediately implies that the sparsity of $N$ is at least $|E(N)|/n  =|E(N_1)|/|S_0| \geq \Omega_d(\eps^{(1-d)/2})$.

The MST of $S'$ consists of $k$ translated copies of $\text{MST}(S)$ and
$k-1$ edges of weight 3 between consecutive copies. That is,
$\|\MST(S')\|=k\, \|\text{MST}(S)\|+3(k-1)=\Theta_d(k\eps^{1-d/2})$.
It follows that the lightness of $N'$ is $\Omega_d(\eps^{-d/2})$, as claimed.
\end{proof}

\section{Upper Bound in the Plane: Reduction to Directional Spanners}
\label{sec:redux}

In this section, we present our general strategy for the proof of Theorem~\ref{thm:upper}, and reduce the construction of a light $(1+\eps)$-spanner for a point set $S$ in the plane to a special case of \emph{directional} spanners for
a point set on the boundary of faces in a (modified) \emph{window partition}.

\subparagraph{Directional $(1+\eps)$-Spanners.}
Our strategy to construct a $(1+\eps)$-spanner for a point set $S$ is to partition the interval of directions $[0,\pi)$ into $O(\eps^{-1/2})$ intervals, each of length $O(\eps^{1/2})$. For each interval $D\subset [0,\pi)$, we construct a geometric graph that serves point pairs $\{a,b\}\subseteq S$ with $\mathrm{dir}(ab)\in D$. Then the union of these graphs over all $O(\eps^{-1/2})$ intervals will serve all point pairs $\{a,b\}\subseteq S$. The following definition formalizes this idea.

\begin{definition}[Directional spanner]
Let $S$ be a finite point set in the plane, $D\subset [0,\pi)$ a set of directions, and $\eps>0$. A geometric graph $G$ is a \emph{directional $(1+\eps)$-spanner} for $S$ and $D$ if $G$ contains an $ab$-path of weight at most $(1+\eps)\|ab\|$ for every $a,b\in S$ with $\mathrm{dir}(ab)\in D$.
\end{definition}

\subparagraph{Reduction to Tilings.}
Assume that we wish to construct a directional $(1+\eps)$-spanner for a set $S$ of $n$ points in the plane and an interval $D=[\frac{\pi-\sqrt{\eps}}{2},\frac{\pi+\sqrt{\eps}}{2}]$ of nearly vertical directions. Our general strategy is the following two-step process: (1) Subdivide a bounding box of $S$ into a collection $\mathcal{F}$ of weakly simple polygons (\emph{faces}) such that no point in $S$ lies in the interior any face. (2) For each face $F\in \mathcal{F}$, construct a directional $(1+\eps)$-spanner $G_F$ for a finite point set $S_F$ on the boundary $\partial F$ of $F$. Specifically, let $S_F$ be the union of  $S\cap \partial F$ and all points where a segment $ab$, with $a,b\in S$ and $\dir(ab)\in D$, crosses an edge of $F$ or contains a vertex of $F$. 
It is easily seen that this construction yields a directional $(1+\eps)$-spanner for $S$.

\begin{lemma}\label{lem:retiling}
Let $S\subset \mathbb{R}^2$, $D$, $\mathcal{F}$, and $S_F$ for all $F\in \mathcal{F}$ as defined above. For each face $F\in \mathcal{F}$, let $G_F$ be a geometric graph that contains, for all $p,q\in S_F$ with $\dir(pq)\in D$ and $pq\subset F$, a $pq$-path of weight at most $(1+\eps)\|pq\|$.
Then $G=\bigcup_{F\in \mathcal{F}} G_F$ is a directional $(1+\eps)$-spanner for $S$.
\end{lemma}
\begin{proof}
Let $a,b\in S$ with $\dir(ab)\in D$. The segment $ab$ is contained in the convex hull of $S$, which is in turn contained in the union of faces in $\mathcal{F}$.
The boundaries of the faces in $\mathcal{F}$ subdivide the line segment $ab$ into a path $(v_0,v_1,\ldots  ,v_m)$ of collinear segments, each of which lies in some face $F\in \mathcal{F}$ with both endpoints in $S_F$. For each $i=1,\ldots , m$, graph $G$ contains a $v_{i-1}v_i$-path of weight at most $(1+\eps)\|v_{i-1}v_i\|$. The concatenation of these paths
is an $ab$-path of weight at most $\sum_{i=1}^m (1+\eps)\|v_{i-1}v_i\| =(1+\eps)\|ab\|$, as required.
\end{proof}

\subparagraph{Remark: How to Tile?}
We need to construct a tiling $\mathcal{F}$ for $S$ that allows us to control the total weight of the spanner $\|G\|=\sum_{F\in \mathcal{F}} \|G_F\|$.
An obvious lower bound for the spanner weight is the weight of the tiling, which is the sum of the perimeters of the faces  $\sum_{F\in \mathcal{F}} \per(F)$.
Let $B$ be an axis-aligned bounding box of $S$.
The simplest tiles would be rectangles, convex polygons, or possibly orthogonally convex polygons (bounded by four staircase paths).
However, the minimum weight of a rectangulation for $n$ points in $B$
is $O(\|\MST(S)\|\log n)$, and this bound is the best possible~\cite{BergK94}. The same bound holds for the minimum weight tiling of $B$ into
orthogonally convex polygons. The minimum weight of a convex partition for $n$ points in $B$ is $O(\|\MST(S)\|\log n/\log \log n)$, and this bound is also tight~\cite{DumitrescuT11}. Due to the logarithmic factors, these tilings would be too heavy for our purposes, for the construction of a $(1+\eps)$-spanner of weight $O(\eps^{-1}\|\MST(S)\|)$.
Instead, we start with a histogram decomposition of weight $O(\|MST(S)\|)$.

\subparagraph{Histogram Decomposition.}
In Section~\ref{sec:hist}, we modify the standard window partition algorithm and tile a bounding box of $S$ with \emph{tame histograms} and \emph{thin histograms}, that we define here. Refer to Fig.~\ref{fig:histogram1}.

\begin{definition}
\begin{itemize}
\item[]
\item A polygon $P=(v_0,v_1,\ldots ,v_m)$ is \emph{simple} if its boundary is a Jordan curve (i.e., the image of an injective map $\gamma:\mathbb{S}^2\rightarrow \mathbb{R}^2$); and $P$ is \emph{weakly simple} if for every $\delta>0$, there exists a simple polygon $P'=(v_0',v_1',\ldots ,v_m')$ with $\|v_iv_i'\|\leq \delta$ for all $i=0,1,\ldots , m$. (Intuitively, this means that $\gamma$ may have self-intersections, but no self-crossings.)
%
%
\item  An \emph{$x$-monotone histogram} is a rectilinear weakly simple polygon bounded by a horizontal line segment and an $x$-monotone path $L$. Similarly,
a \emph{$y$-monotone histogram} is a rectilinear weakly simple polygon bounded by a vertical line segment and a $y$-monotone path $L$.
\item An $x$-monotone (resp., $y$-monotone) histogram is
    \emph{$\tau$-tame} for $\tau>0$ if for every horizontal (resp., vertical) chord $ab$, with $a,b\in L$, the subpath $L_{ab}$ of $L$ between $a$ and $b$ satisfies $\|L_{ab}\|\leq (1+\tau)\|ab\|$.
\item Finally, a \emph{tame histogram} is an $x$-monotone $1$-tame histogram, and a \emph{thin histogram} as a $y$-monotone $\eps^{1/2}$-tame histogram.
\end{itemize}
\end{definition}

\begin{figure}[htbp]
 \centering
 \includegraphics[width=0.95\textwidth]{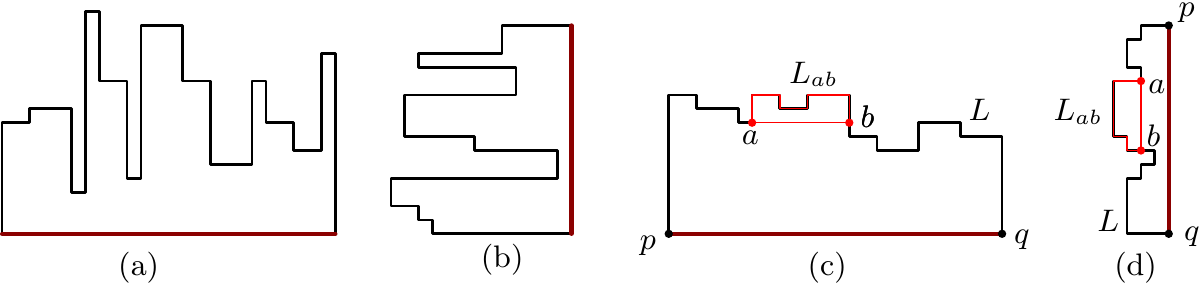}
 \caption{
 (a) An $x$-monotone histogram.
 (b) A $y$-monotone histogram.
 (c) A tame histogram.
 (d) A thin histogram.}
    \label{fig:histogram1}
\end{figure}

For a tame or thin histogram $P$ and a finite set $S$ of points on the boundary of $P$, we construct geometric graphs that achieve the stretch factor $1+\eps$ for all point pairs $a,b\in S$ such that $\dir(ab)\in [\frac{\pi-\sqrt{\eps}}{2},\frac{\pi+\sqrt{\eps}}{2}]$ \emph{and} $ab$ is a chord of $P$. A line segment $ab$ is a \emph{chord} of a weakly simple polygon $P$ if $a,b\in \partial P$, and $ab\subset P$. The \emph{perimeter} of a weakly simple polygon $P$, denoted $\per(P)$, is the total weight of the edges of a closed polygonal chain on the boundary of $P$; and the \emph{horizontal perimeter}, denoted $\hper(P)$, is the total weight of the horizontal edges of that chain.
Note that $\hper(H)=2\,\wdth(P)$ for an $x$-monotone histogram $H$; and $\hper(H)\geq 2\,\wdth(P)$ for a $y$-monotone histogram.

\begin{restatable}{lemma}{tilinglemma}\label{lem:tiling}
We can subdivide a (weakly) simple rectilinear polygon $P$ into a collection $\mathcal{F}$ of tame histograms and thin histograms such that 
$\sum_{F\in \mathcal{F}} \per(F) \leq O(\eps^{-1/2}\per(P))$
and $\sum_{F\in \mathcal{F}} \hper(F) \leq O(\per(P))$.
\end{restatable}

\begin{restatable}{lemma}{histlemma}\label{lem:hist}
Let $F$ be a tame or thin histogram, $S\subset \partial F$ a finite point set, $\eps\in (0,1]$, and $D=[\frac{\pi-\sqrt{\eps}}{2},\frac{\pi+\sqrt{\eps}}{2}]$ an interval of nearly vertical directions. Then there exists a geometric graph $G$ of weight $O(\per(F)+\eps^{-1/2}\, \hper(F))$ such that for all $a,b\in S$, if $ab$ is a chord of $F$ and $\dir(ab)\in D$, then $G$ contains an $ab$-path of weight at most $(1+O(\eps))\|ab\|$.
\end{restatable}

We prove Lemma~\ref{lem:tiling} in Section~\ref{sec:hist} and Lemma~\ref{lem:hist} in Section~\ref{sec:thin}.
In the remainder of this section, we show that these lemmas imply Lemma~\ref{lem:dir}, which in turn implies Theorem~\ref{thm:UB}.

\subparagraph{Stretch Factor of $1+\eps$ Versus $1+O(\eps)$.}
In the geometric spanners we construct, an $st$-path may comprise $O(1)$ subpaths, each of which is angle-bounded or contained in an \textsf{SLT}.
For the ease of presentation, we typically establish a stretch factor of $1+O(\eps)$ in our proofs. It is understood that $1+\eps$ can be achieved by a suitable scaling by a constant factor.

\begin{lemma}\label{lem:dir}
Let $S\subset \mathbb{R}^2$ be a finite point set, $\eps\in (0,1]$, and $D\subset [0,\pi)$ an interval of length $\sqrt{\eps}$. Then there exists a \emph{directional $(1+\eps)$-spanner} for $S$ and $D$ of weight $O(\eps^{-1/2}\, \|\MST(S)\|)$.
\end{lemma}
\begin{proof}
We may assume, by applying a suitable rotation, that $D=[\frac{\pi-\sqrt{\eps}}{2},\frac{\pi+\sqrt{\eps}}{2}]$, that is, an interval of nearly vertical directions. We construct a directional $(1+\eps)$-spanner for $S$ and $D$ of weight $O(\eps^{-1/2}\, \|\MST(S)\|)$.

Assume w.l.o.g.\ that the unit square $U=[0,1]^2$ is a minimum axis-parallel bounding square of $S$. In particular, $S$ has two points on two opposite sides of $U$, and so $1\leq \diam(S)\leq \|\MST(S)\|$. Our initial graph $G_0$ is composed of the boundary of $U$ and a \emph{rectilinear MST}\footnote{A \emph{rectilinear minimum spanning tree} of a finite set $S$ in $\mathbb{R}^d$ is a Steiner tree for $S$ composed of axis-parallel edges and having minimum weight in $L_1$-norm.} of $S$, where $\|G_0\|=O(\|\MST(S)\|)$.
Since each edge of $G_0$ is on the boundary of at most two faces, the total perimeter of all faces of $G_0$ is also $O(\|\MST(S)\|)$.
Lemma~\ref{lem:tiling} yields subdivisions of the faces of $G_0$ into a collection $\mathcal{F}$ of tame or thin histograms with  $\sum_{F\in \mathcal{F}}\per(F)=O(\eps^{-1/2}\|\MST(S)\|)$
and $\sum_{F\in \mathcal{F}}\hper(F)=O(\|\MST(S)\|)$,

Let $K(S)$ be the complete graph induced by $S$. For each face $F\in \mathcal{F}$, let $S_F$ be the union of $S\cap \partial F$ and the set of all points where an edge of $K(S)$ crosses an edge of $F$ or passes through a vertex of $F$. For each face $F$, Lemma~\ref{lem:hist} yields a geometric graph $G_F$ of weight $O(\per(F)+\eps^{-1/2}\hper(F))$ with respect to the finite point set $S_F\subset \partial F$.

We can now put the pieces back together.
Let $G$ be the union of $G_0$ and the graphs $G_F$ for all $F\in \mathcal{F}$.
By Lemma~\ref{lem:retiling}, $G$ is a directional $(1+\eps)$-spanner for $S$ and $D$. The weight of $G$ is bounded by
$\|G\|=\|G_0\|+\sum_{F\in \mathcal{F}} \|G_F\|
= O(\|\MST(S)\|+\sum_{F\in \mathcal{F}} (\per(F)+\eps^{-1/2}\hper(F)))
= O(\eps^{-1/2}\|\MST(S)\|)$.
\end{proof}

We prove Theorem~\ref{thm:upper} in the following form.

\begin{theorem}\label{thm:UB}
For every finite point sets $S\subset \mathbb{R}^2$ and $\eps\in (0,1]$, there exists a Euclidean Steiner $(1+\eps)$-spanner of weight $O(\eps^{-1}\,\|\MST(S)\|)$.
\end{theorem}

\begin{proof}[Proof of Theorem~\ref{thm:UB}]
Let $S$ be a finite set in the plane, let $\eps\in (0,1]$, and put $k=\lceil \pi\eps^{-1/2}\rceil$. Partition the space of directions into $k$ intervals of equal length, as  $[0,\pi)=\bigcup_{i=1}^{k} D_i$.
By Lemma~\ref{lem:dir}, there exists a directional $(1+\eps)$-spanner of weight $O(\eps^{-1/2}\|\MST(S)\|)$ for $S$ and $D_i$ for every  $i\in \{1,\ldots , k\}$.
Let $G=\bigcup_{i=1}^{k} G_i$. For every point pair $s,t\in S$,
we have $\mathrm{dir}(st)\in D_i$ for some $i\in \{1,\ldots,k\}$, and $G_i\subset G$ contains an $st$-path of weight at most $(1+\eps)\|st\|$. Consequently, $G$ is a Euclidean Steiner $(1+\eps)$-spanner for $S$.
The weight of $G$ is
$\|G\|\leq \sum_{i=1}^{k} \|G_i\|
\leq {\lceil\pi\eps^{-1/2}\rceil}\cdot O(\eps^{-1/2}\|\MST(S)\|)
\leq O(\eps^{-1}\|\MST(S)\|)$, as required.
\end{proof}

\section{Construction of a Tiling}
\label{sec:hist}

The so-called \emph{window partition} of rectilinear simple polygon $P$ is a recursive subdivision of $P$ into histograms~\cite{Edelsbrunner1984167,Levcopoulos,Link00,Suri90}. It can be computed in $O(n\log n)$ time if $P$ has $n$ vertices. It was originally designed for data structures that support orthogonal visibility and minimum-link path queries in $P$. Importantly, every axis-parallel chord of $P$ intersects (\emph{stabs}) at most three histograms of the decomposition. The stabbing property implies that the total perimeter of the histograms is $O(\per(P))$. For completeness, we present the standard window partition and the weight analysis here.

\begin{figure}[htbp]
 \centering
 \includegraphics[width=0.8\textwidth]{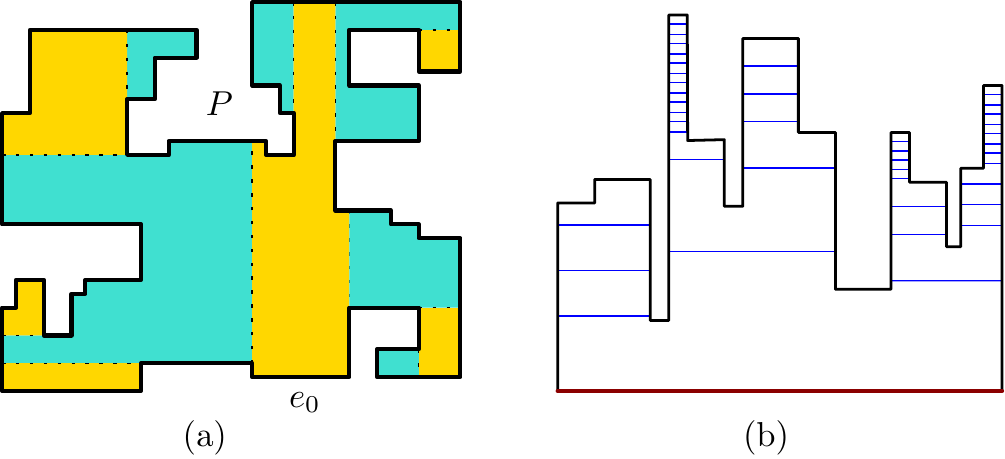}
 \caption{
 (a) The standard window partition of a rectilinear polygon $P$ into histograms, starting from a horizontal edge $e_0$
 (b) A decomposition of an $x$-monotone histogram into tame histograms.}
    \label{fig:histogram2}
\end{figure}

\begin{lemma}\label{lem:window}
Every rectilinear (weakly) simple polygon $P$ can be subdivided into a collection $\mathcal{H}$ of histograms such that
$\sum_{H\in \mathcal{H}}\per(H)\leq O(\per(P))$.
\end{lemma}
\begin{proof}
We describe a recursion on instances $(R,e)$, where $R$ is a rectilinear (weakly) simple polygon, and $e$ is an edge of $R$.
We also maintain the invariant that $\partial R\setminus e\subset \partial P$. Initially, let $R=P$ and $e$ an arbitrary horizontal edge of $P$; clearly $\partial P\setminus e\subset \partial P$. In one iteration, consider an instance $(R,e)$; see Fig.~\ref{fig:histogram2}(a) for example.

Assume w.l.o.g.\ that $(R,e)$ is an instance where $e$ is horizontal edge of $R$. We define a histogram $H\subset R$ as the set of all points that can be connected to a point in $e$ by a vertical line segment in $R$. Let $\mathcal{C}$ be the collection of connected components $C$ of $R\setminus H$. Each component $C\in \mathcal{C}$ is a weakly simple rectilinear polygon that has a unique vertical edge (\emph{window}) $w(C)$ along the boundary of $H$. Recurse on the instances $(C,w(C))$ for all $C\in \mathcal{C}$ if $\mathcal{C}\neq \emptyset$.

\smallskip\noindent\emph{Weight analysis.}
Consider the input rectilinear polygon $P$, and the subdivision created by the above algorithm. Each iteration of the recursion handles an instance $(R,e)$, and creates an $x$- or a $y$-monotone histogram $H_e$, which is not partitioned further. The cost of creating $H_e$ equals to the weight of the windows on boundary between $H_e$ and $R\setminus H_e$. Each component of $\partial H_e\cap \partial(R\setminus H_e)$ is the edge $w(C)$ of an instance $(C,w(C))$ in a recursive call. 

At the next level of the recursion, for each instance $(C,w(C))$, the algorithm constructs a histogram $H_{w(C)}\subset C$. By our invariant, we have $\partial C\setminus w(C)\subset \partial R$; and we charge the weight of $w(C)$ to the boundary of $R$ as follows. When $w(C)$ is a horizontal (vertical) side of $C$, then we charge the orthogonal projection of $w(C)$ to horizontal (vertical) edges of $\partial C\setminus w(C)\subset \partial R$. The projection consists of one or more horizontal (vertical) line segments of total weight $\|w(C)\|$. We charge $w(C)$ to the common boundary of $R$ and $H_{w(C)}$, which is not partitioned further, hence every line segment $s$ on the boundary of the input polygon $P$ receives a charge of at most $\|s\|$, and the overall weight of all windows is bounded by $\per(P)$. As each window contributes to the perimeter of precisely two faces in $\mathcal{C}$, then
$\sum_{C\in \mathcal{C}}\per(C)=O(\per(P))$, as claimed.
\end{proof}

\subparagraph{Subdivision of Histograms into Tame and Thin Histograms}
%
Let $H$ be a histogram produced by the window partition algorithm (Lemma~\ref{lem:window}). We subdivide $H$ into tame or thin  histograms by a sweepline algorithm. Dumitrescu and T\'oth~\cite{DumitrescuT09} used similar methods to partition a histogram into histograms of constant geometric dilation.

\begin{lemma}\label{lem:sweep}
For every $\tau\in (0,1]$, an $x$-monotone histogram $H$ can be subdivided into a collection $\mathcal{T}$ of $\tau$-tame histograms such that
$\sum_{T\in \mathcal{T}}\per(T)=O(\tau^{-1}\,\per(H))$ and
$\sum_{T\in \mathcal{T}}\vper(T)=O(\vper(H))$.
\end{lemma}
\begin{proof}
Let $H$ be an $x$-monotone histogram bounded by a horizontal line segment $pq$ from below, and an $x$-monotone $pq$-path $L$ from above.
We describe a sweepline algorithm that recursively subdivides $H$ with horizontal lines; see Fig.~\ref{fig:histogram2}(b). Initially, set $\mathcal{T}=\emptyset$. Sweep $H$ top-down with a horizontal line $\ell$, and incrementally update $L$, $H$, and $\mathcal{T}$ as follows. Whenever the sweepline $\ell$ contains a chord $ab$ of $H$ such that the subpath $L_{ab}$ of $L$ between $a$ and $b$ has weight $(1+\tau)\|ab\|$, then we add the simple polygon bounded by $ab$ and $L_{ab}$ into $\mathcal{T}$, and replace $L_{ab}$ with the line segment $ab$ in both $L$ and $H$. When the sweepline $\ell$ reaches the base of $H$, we add $H$ to $\mathcal{T}$, and return $\mathcal{T}$.

First note that each polygon added into $\mathcal{T}$ is a $\tau$-tame histogram. By construction, $\sum_{T\in \mathcal{T}}\per(T)$ is proportional to the sum of $\per(H)$ and the total weight of all horizontal chords $ab$ inserted by the algorithm. At the time when the algorithm creates a tame histogram bounded by $ab$ and $L_{ab}$, we can charge the weight $\|ab\|$ of the chord $ab$ to the vertical edges of the path $L_{ab}$. Since the algorithm inserts only horizontal edges, all vertical edges along this path lie on the boundary of the input polygon. Consequently, the total weight of vertical edges of $L_{ab}$ is $\tau\,\|ab\|$; and are not charged in subsequent steps of the recursion.
Each vertical segment $s$ on the boundary of $H$ receives a charge of at most $\tau^{-1}\|s\|$. Overall, the total weight of the edges inserted by the algorithm is $O(\tau^{-1}\,\vper(H))$, hence
$\sum_{T\in \mathcal{T}}\per(T)\leq \per(H)+O(\tau^{-1}\,\vper(H)) \leq O(\tau^{-1}\,\per(H))$, as required.
Since all subdivision edges are horizontal, then
$\sum_{T\in \mathcal{T}}\vper(T)=O(\vper(H))$.
\end{proof}

\subparagraph{Proof of Lemma~\ref{lem:tiling}.}
The combination of Lemmas~\ref{lem:window} and~\ref{lem:sweep} implies Lemma~\ref{lem:tiling}.

\tilinglemma*

\begin{proof}
Let $P$ be a rectilinear (weakly) simple polygon. By Lemma~\ref{lem:window}, we can partition $P$ into a collection
$\mathcal{H}$ of histograms such that
$\sum_{H\in \mathcal{H}}\per(H)=O(\per(P))$.

Using Lemma~\ref{lem:sweep} with $\tau=1$, we can partition each $x$-monotone histogram $H\in \mathcal{H}$ into a collection $\mathcal{T}(H)$ of tame histograms of total perimeter $O(\per(H))$.
This implies that their horizontal perimeter is also bounded by $O(\per(H))$.
Using  Lemma~\ref{lem:sweep} with $\tau=\eps^{1/2}$, we can also partition each $y$-monotone histogram $H\in \mathcal{H}$ into a collection $\mathcal{T}(H)$ of thin histograms of total perimeter $O(\eps^{-1/2}\per(H))$ and total horizontal perimeter $O(\per(H))$.

Overall, we obtain a subdivision of $P$ into a collection $\mathcal{F}=\bigcup_{H\in \mathcal{H}} \mathcal{T}(H)$ of histograms, each of which is either tame or thin, of total perimeter $O(\eps^{-1/2}\per(P))$ and total horizontal perimeter $O(\per(P))$, as required.
\end{proof}

\section{Generalized Shallow Light Trees}
\label{sec:SLT}

Shallow-light trees (\textsf{SLT}) were introduced by Awerbuch et al.~\cite{AwerbuchBP90}
and Khuller et al.~\cite{KhullerRY93}. Given a source $s$ and a point set $S$ in a metric space, an \emph{$(\alpha,\beta)$}-\textsf{SLT} is a Steiner tree rooted at $s$ that contains a path of weight at most $\alpha \,\|st\|$ between the \emph{source} $s$ and any point $t\in S$, and has weight at most $\beta\,\|\MST(S)\|$. We build on the following basic variant of \textsf{SLT} between a source $s$ and a set $S$ of collinear points in the plane; see Fig.~\ref{fig:SLT}.

\begin{figure}[htbp]
 \centering
 \includegraphics[width=0.2\textwidth]{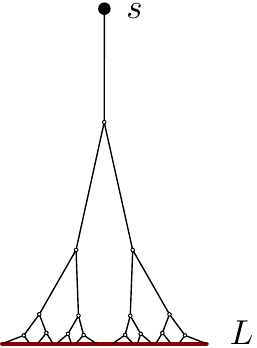}
 \caption{A shallow-light tree for a source $s$ and a set of collinear points on a line segment $L$. 
    }
    \label{fig:slt}
\end{figure}

\begin{lemma}[Solomon~{\cite[Section~2.1]{Solomon15}}]\label{lem:shallow}
Let $\eps\in (0,1]$, let $s=(0,\eps^{-1/2})$ be a point on the $y$-axis, and let $S$ be a set of points in the line segment $L=[-\frac12,\frac12]\times \{0\}$ in the $x$-axis. Then there exists a geometric graph of weight $O(\eps^{-1/2})$ that contains, for every point $t\in L$, an $st$-path $P_{st}$ with $\|P_{st}\|\leq (1+\eps)\,\|st\|$.
\end{lemma}

We note that the weight analysis of the $st$-path $P_{st}$ in an \textsf{SLT} does not use  angle-boundedness. In particular, an \textsf{SLT} may contain short edges of arbitrary directions close to $t$, but the long edges are nearly vertical.
%

In Section~\ref{ssec:stairs}, we generalize Lemma~\ref{lem:shallow}, and construct shallow-light trees between a source $s$ and points on a staircase path. In Section~\ref{ssec:combination}, we show how to combine two shallow-light trees to obtain a spanner between point pairs on two staircase paths.

\subsection{Shallow-Light Trees for Staircase Chains}
\label{ssec:stairs}

We present a new, slightly modified proof for Solomon's result on \textsf{SLT}s between a single source $s$ and a horizontal line segment, and then adapt the modified proof to obtain a \textsf{SLT} between $s$ and an $x$- and $y$-monotone polygonal chain. In the proof below, we use the Taylor estimates $\cos x\geq 1-x^2/2$ and $\sin x\geq x/2$ for $x\leq \pi/3$.

\begin{figure}[htbp]
 \centering
 \includegraphics[width=0.75\textwidth]{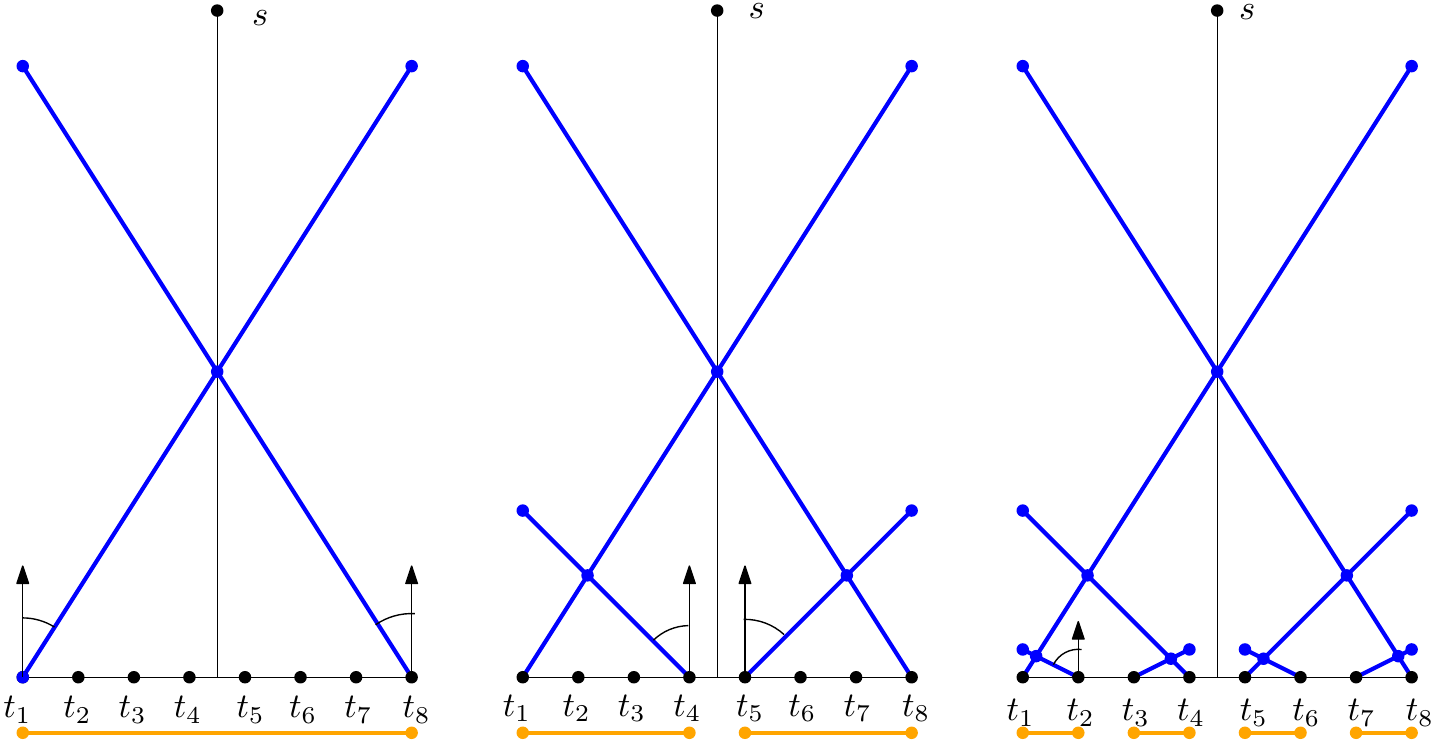}
 \caption{The segments added to graph $G$ at level $j=0,1,2$ for $m=2^3=8$ points.
 The intervals $[t_a,t_b]$ at level $j$ are indicated below the line $L$.}
\label{fig:SLT}
\end{figure}

\begin{proof}[Alternative proof for Lemma~\ref{lem:shallow}]
Assume w.l.o.g.\ that $\eps=2^{-k}$ for $k\in \mathbb{N}$. Let $T=\{t_i: i=1,\ldots , 2^{k+1}\}$ be $2^{k+1}$ points on the line segment $L=[-\frac12,\frac12]\times \{0\}$ with uniform $1/(2^{k+1}-1)< \eps$ spacing between consecutive points.
Consider the standard binary partition of $\{1,\ldots , 2^{k+1}\}$ into intervals, associated with a binary tree: At level 0, the root corresponds to the interval $[1,2^{k+1}]$ of all $2^{k+1}$ integers. At level $j$, we have intervals $[i\cdot 2^{k-j}+1,(i+1)\cdot 2^{k-j}]$ for $i=0,\ldots  ,2^j-1$. Note that if a point $q$ is the left (resp., right) endpoint of an interval at a level $j$, then $q$ is the left (resp., right) endpoint of all descendant intervals that contains $q$.

For every $q\in \{1,\ldots , 2^{k+1}\}$, we define a line segment $\ell_q$ with one endpoint at $t_q$: Let $j\geq 0$ be the smallest level such that $q$ is an endpoint of some interval $I_q$ at level $j$. Let $T_q$ be the line segment along the $x$-axis spanned by the points corresponding to $I_q$. If $q$ is the left (resp., right) endpoint of $I_q$, then let $\ell_q$ be the line segment of direction $\frac{\pi}{2}-2^{(j-k)/2}$ (resp., $\frac{\pi}{2}+2^{(j-k)/2})$
such that its orthogonal projection to the $x$-axis is $T_q$; see Fig.~\ref{fig:SLT}.
Note that for $j=0$, we use directions $\frac{\pi}{2}\pm 2^{-k/2}=\frac{\pi}{2}\pm \sqrt{\eps}$.
Let $G$ be the union of segments $\ell_q$ for $q=1,\ldots , 2^{k+1}$,
the horizontal segment $L$, and the vertical segment from $s$ to the origin.

\smallskip\noindent\emph{Lightness analysis.}
We show that $\|G\|=O(\eps^{-1/2})$. We have $\|L\|=1$, and the length of the vertical segment between $s$ and the origin is $\eps^{-1/2}$. At level $j$ of the binary tree, we construct $2^j$ segments $\ell$, each of length $\|\ell\|\leq 2^{-j}/\sin(2^{(j-k)/2})\leq 2\cdot 2^{(k-3j)/2}$. Summation over all levels yields
$\sum_{j=0}^{k} 2^{j}\cdot 2\cdot 2^{(k-3j)/2}= 2^{k/2}\cdot 2\cdot \sum_{j=0}^k 2^{-j/2}= O(2^{k/2}) = O(\eps^{-1/2})$.

\smallskip\noindent\emph{Source-stretch analysis.}
We show that $G$ contains an $st_q$-path of length $(1+O(\eps))\|st_q\|$ for all $q=1,\ldots, 2^{k+1}$. First note that $\|st_q\|\geq |y(s)-y(t_q)|= \eps^{-1/2}$.
For each interval $[t_a,t_b]$ in the binary tree, $\ell_a$ and $\ell_b$ have positive and negative slopes, respectively, and so they cross above the interval $[t_a,t_b]$. Consequently, for every point $t_q$, the union of the $k+1$ segments corresponding to the intervals that contain $t_q$ must contain a $y$-monotonically increasing path $P_q$ from $t_q$ to $s$. The $y$-projection of this path has length $\eps^{-1/2}$. Consider one edge $e$ of $P_q$ along a segment $\ell$ at level $j$, which has direction $\frac{\pi}{2}\pm \alpha=\frac{\pi}{2}\pm 2^{(j-k)/2}$. The difference between the length of $e$ and the $y$-projection of $e$ is
\begin{align*}
\|e\|(1-\cos \alpha)
 &\leq \|\ell\|(1-\cos\alpha)
 \leq 2^{-j} \frac{1-\cos\alpha}{\sin \alpha}
 \leq 2^{-j} \frac{\alpha^2/2}{\alpha/2}\\
  &= 2^{-j} \alpha
  = 2^{-j}\cdot 2^{(j-k)/2} =2^{-(j+k)/2}.
\end{align*}
Since $P_q$ contains at most one edge in each level, summation over all edges of $P_q$ yields
\[
\sum_{j=0}^{k} 2^{-(j+k)/2}
= 2^{-k/2} \sum_{j=0}^{k} 2^{-j/2}
= O(\eps^{1/2})\leq \|st_q\|\cdot O(\eps).
\]

Finally, for an arbitrary point $t\in L$, we have $\|st\|\geq |y(s)-y(t)|=\eps^{-1/2}$, and $G$ contains an $st$-path that consists of an $st_q$-path from $s$ to the point $t_q$ closest to $t$, followed by the horizontal segment $t_qt$ of weight $\|t_qt\|< 1/2^k\leq \eps$. The total weight of this path is $(1+O(\eps))\|st\|$.
After suitable scaling of the constant coefficients, $G$ contains a path of weight at most $(1+\eps)\|st\|$ for any $t\in L$, as required.
\end{proof}

Note that we have shown that the graph $G$ contains an $st$-path $P_{st}$ with $\|P_{st}\|\leq (1+\eps)\eps^{-1/2}$ for every point $t\in L$; and combined this upper bound with the trivial lower bound $\eps^{-1/2}=|y(s)-y(t)|\leq \|st\|$. We are now ready to generalize Lemma~\ref{lem:shallow} to staricases.

\begin{lemma}\label{lem:stairs}
Let $\eps\in (0,1]$, let $s=(0,\eps^{-1/2})$ be a point on the $y$-axis, and let $L$ be an $x$- and $y$-monotone increasing staircase path between the vertical lines $x=\pm \,\frac12$, such that the right endpoint of $L$ is $(\frac12,0)$ on the $x$-axis. Then there exists a geometric graph $G$ comprised of $L$ and additional edges of weight $O(\eps^{-1/2})$ such that $G$ contains, for every $t\in L$, an $st$-path $P_{st}$ with
$\|P_{st}\|  \leq (1+O(\eps))\,|y(s)-y(t)| \leq (1+O(\eps))\,\|st\|$.
\end{lemma}
We can adjust the construction above as follows; refer to Fig.~\ref{fig:SLT2}.
\begin{proof}
Assume w.l.o.g.\ that $\eps=2^{-k}$ for some $k\in \mathbb{N}$.
Let $T=\{t_i: i=1,\ldots , 2^{k+1}\}$ be $2^{k+1}$ points in $L$ on equally spaced vertical lines, with spacing $1/(2^{k+1}-1)< \eps$.
Consider the standard binary partition of  $\{1,\ldots , 2^{k+1}\}$ into intervals as in the previous proof.

\begin{figure}[htbp]
 \centering
 \includegraphics[width=0.8\textwidth]{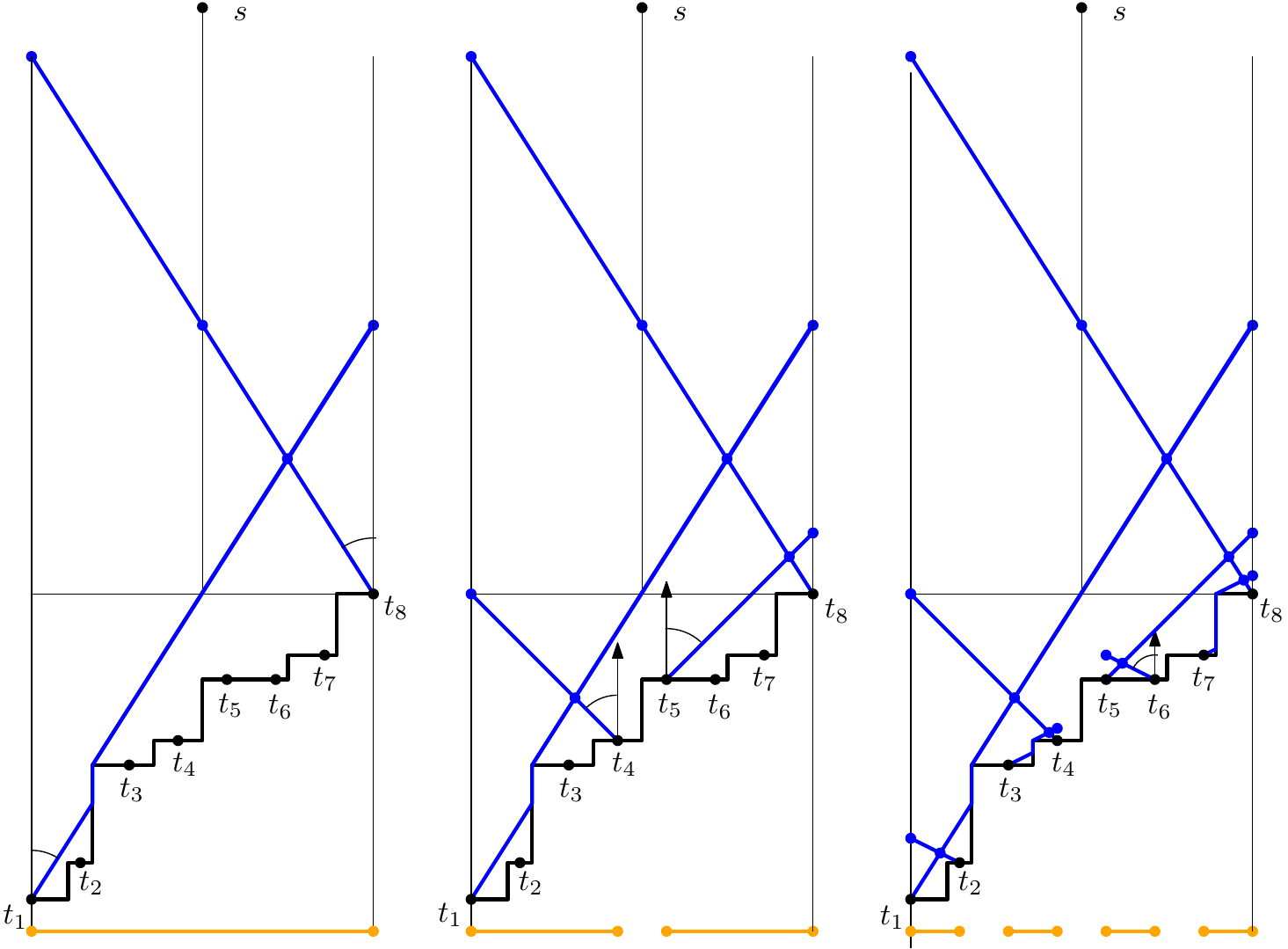}
 \caption{The paths $\gamma_q$ added to graph $G$ at level $j=0,1,2$ for $m=2^3=8$ points.
 The intervals $[t_a,t_b]$ at level $j$ are indicated below the staircase path $L$.}
    \label{fig:SLT2}
\end{figure}

For every $q\in \{1,\ldots , 2^{k+1}\}$, we define a polygonal path $\gamma_q$ with one endpoint at $t_q$; see Fig.~\ref{fig:SLT2}. Let $j\geq 0$ be the smallest level such that $t_q$ is an endpoint of some interval
$I_q$ at level $j$. If $t_q$ is the right endpoint of $I_q$, then let $\gamma_q$
be the line segment of direction $\frac{\pi}{2}+2^{(j-k)/2}$ such that its
$x$-projection is $T_q$. If $t_q$ is the left endpoint of $I_q$, then $\gamma_q$ will be an $x$- and $y$-monotone path whose $x$-projection is $T_q$, and its edges will be vertical segments along $L$ and segments of direction $\alpha_q=\frac{\pi}{2}-2^{(j-k)/2}$. Specifically, $\gamma_q$ starts from $t_q$ with a line of direction $\alpha_q$. Whenever $\gamma_q$ encounters a vertical edge of $L$, it follows it upward until its upper endpoint, and then continues in direction $\alpha_q$.

Let $G$ be the union of all paths $\gamma_q$ for $q=1,\ldots , 2^{k+1}$, as well as the path $L$, and the vertical segment from $s$ to the origin. This completes the construction of $G$.

\smallskip\noindent\emph{Lightness analysis.}
We show that $\|G\|=\|L\|+O(\eps^{-1/2})$. The distance between $s$ and $L$ is $\eps^{-1/2}$. For every $q\in \{1,\ldots, 2^{k+1}\}$, the path $\gamma_q$ is composed of vertical segments along $L$, and nonvertical segments whose total weight is the same as $\|\ell_q\|$ in the proof of Lemma~\ref{lem:shallow}, where we have seen that
$\sum_{q=1}^{2^{k+1}} \|\ell_q\| = O(\eps^{-1/2})$. Consequently, $\|G\|=\|L\|+O(\eps^{-1/2})$.

\smallskip\noindent\emph{Source-stretch analysis.}
We show that $G$ contains an $st_q$-path of weight $(1+O(\eps))\|st_q\|$ for all $q=1,\ldots, 2^{k+1}$. Denoting $y(t_q)$ the $y$-coordinate of point $t_q$, we have $\|st_q\|\geq |y(s)-y(t_q)|=\eps^{-1/2}+|y(t_q)|$.
For each interval $[t_a,t_b]$ in the binary tree, the paths $\gamma_a$ and $\gamma_b$ cross above the portion of $L$ between $t_a$ and $t_b$. Consequently, for every point $t_q$, the union of the $k+1$ paths $\gamma$ corresponding to the intervals that contain $t_q$ must contain a $y$-monotonically increasing path $P_q$ from $t_q$ to $s$.
The $y$-projection of this path has weight  $|y(s)-y(t_q)|\eps^{-1/2}+|y(t_q)|$.
Some of the edges of this path may be vertical.
Consider the union of all nonvertical edges $e$ of $P_q$ along a path $\gamma$ at level $j$,
which all have direction $\frac{\pi}{2}\pm 2^{(j-k)/2}$. The difference
between the length of $e$ and the $y$-projection of $e$ is bounded by the same analysis as in the proof of Lemma~\ref{lem:shallow}. Summation over all levels yields $O(\eps^{1/2})\leq \|st_q\|\cdot O(\eps)$.

Finally, for an arbitrary point $t\in L$, we have $\|st\|\geq |y(s)-y(t)|=\eps^{-1/2}+|y(t)|$, and $G$ contains an $st$-path $P_{st}$
comprising an $st_q$-path from $s$ to the closest point $t_q$ to the right of $t$, followed by an $x$- and $y$-monotone path along $L$ in which the total length of the horizontal edges is bounded by $1/2^k\leq \eps$ (and the length of vertical segments might be arbitrary). The vertical segments between $t_q$ and $t$ do not contribute to the error term $\|st\|-|y(s)-y(t)|$. The analysis in the proof of Lemma~\ref{lem:shallow}
yields $\|P_{st}\|-|y(s)-y(t)|
\leq O(\sqrt{\eps})
\leq O(\eps)\, |y(s)-y(t)|$.
Hence 
$\|P_{st}\|  
\leq (1+O(\eps))\,|y(s)-y(t)| 
\leq (1+O(\eps))\,\|st\|$, as required.
\end{proof}

\subsection{Shallow-Light Trees for $y$-Monotone Chains}
\label{ssec:monotone}

We further generalize Lemma~\ref{lem:stairs}, and construct an SLT between a source $s$ and a $y$-monotone rectilinear path $L$ at distance $\eps^{-1/2}$ from $s$.

\begin{lemma}\label{lem:monotone}
Let $\eps\in (0,1]$, let $s=(0,\eps^{-1/2})$ be a point on the $y$-axis, and let $L$ be a $y$-monotone rectilinear path such that the top endpoint of $L$ is on the $x$-axis, the total weight of the horizontal edges in $L$ is at most 1, and $L$ lies in the vertical strip between the lines $x=\pm \,\frac12$. Then there exists a geometric graph $G$ comprised of $L$ and additional edges of weight $O(\eps^{-1/2})$ such that $G$ contains, for every $t\in L$, an $st$-path $P_{st}$ with
$\|P_{st}\|  \leq (1+O(\eps))\,|y(s)-y(t)| \leq (1+O(\eps))\,\|st\|$.
\end{lemma}
We reduce the case of $y$-monotone paths to staircase paths.
\begin{proof}
Let $L=(u_0,\ldots, u_m)$ be a $y$-monotone increasing rectilinear path, where the total weight of horizontal edges is at most 1; refer to Fig.~\ref{fig:SLT3}. Let $L'=(v_0,\ldots ,v_m)$ be a corresponding staircase path, which is both $x$- and $y$-monotone increasing, such that $v_m=u_m$ (i.e., their top endpoints are the same)
and for all $i=1,\ldots ,m$, the edges $u_{i-1}u_i$ and $v_{i-1}v_i$ are parallel and have the same length. (However, if $u_{i-1}u_i$ is a horizontal and $x$-monotone decreasing, then $v_{i-1}v_i$ is $x$-monotone increasing.)
Note that $y(u_i)=y(v_i)$ for all $i=0,\ldots ,m$.

\begin{figure}[htbp]
 \centering
 \includegraphics[width=0.9\textwidth]{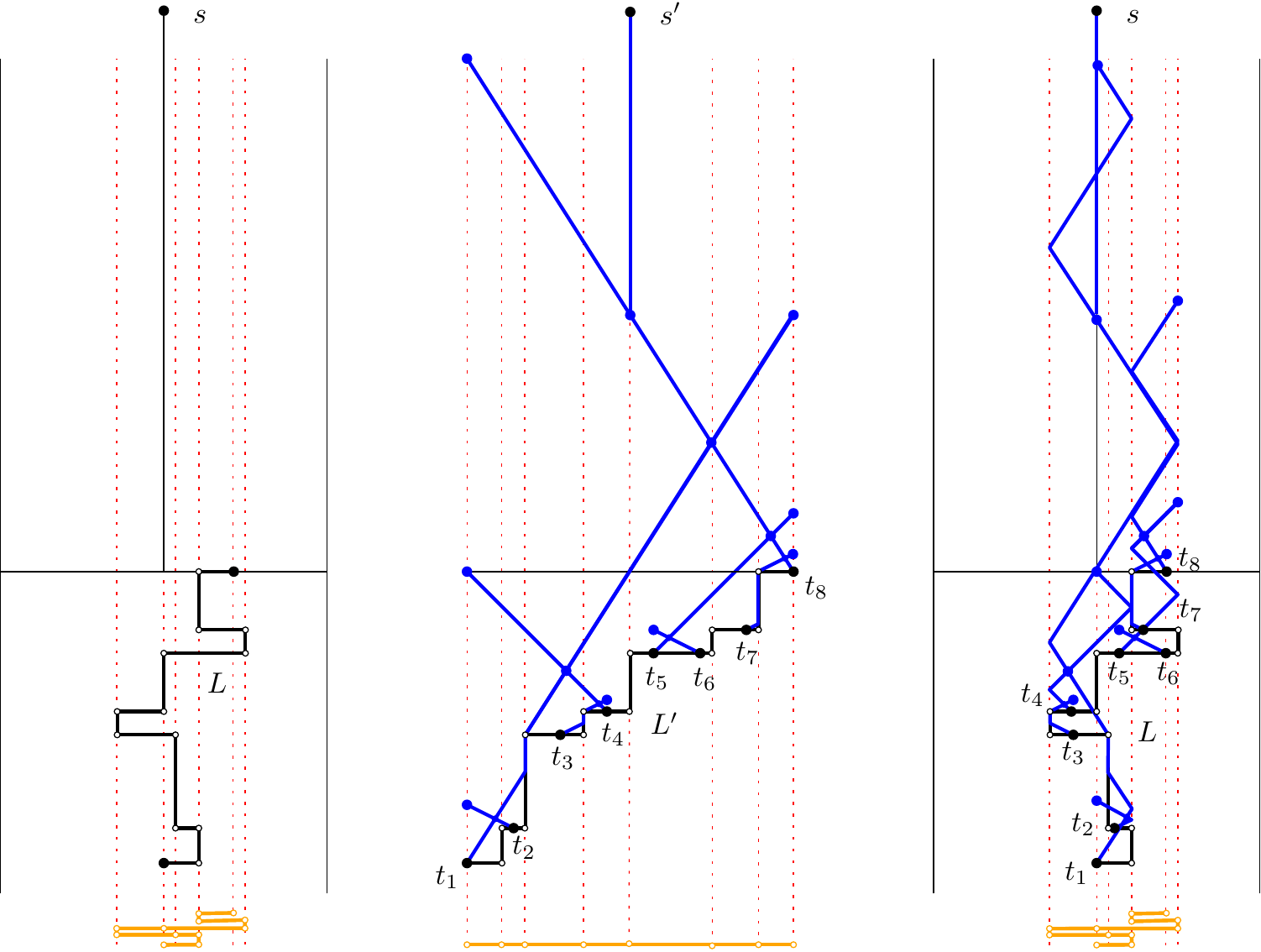}
 \caption{(a) The $y$-monotone rectilinear path $L$ and a source $s$.
 (b) The staircase path $L'$, and a single-source spanner $G'$ for $s'$ and $L'$.
 (c) A spanner $G=\varrho(G')$ for $s=\varrho(s')$ and $L=\varrho(L')$.}    \label{fig:SLT3}
\end{figure}

Let $R$ and $R'$ be the minimal vertical strips that contain $L$ and $L'$, respectively. We define a map $\varrho:R'\rightarrow R$ such that $\varrho(L')=L$. For $i=1,\ldots ,m$, let $R_i$ be the vertical strip bounded by the vertical lines passing through $u_{i-1}$ and $u_i$. Similarly, let $R_i'$ be the vertical strip bounded by the vertical lines passing through $v_{i-1}$ and $v_i$. Then $R=\bigcup_{i=1}^m R_i$ and $R'=\bigcup_{i=1}^m R_i'$. For every $i$, there is a unique isometry $\varrho_i:R_i'\rightarrow R_i$, composed of a translation by $\overrightarrow{v_iu_i}$ and a possible reflection in a vertical line, such that $\varrho(u_{i-1})=v_{i-1}$ and $\varrho(u_{i})=v_{i}$. The isometries $\varrho_1,\ldots , \varrho_m$ jointly define a  map $\varrho:R'\rightarrow R$. Since $\varrho:R'\rightarrow R$ is surjective, there exists a point $s'\in R'$ with $\varrho(s')=s$.

Note that $\varrho$ is a contraction, that is, $\|\varrho(a)\varrho(b)\|\leq \|ab\|$ for all $a,b\in R'$, as it maintains the $y$-coordinates of $a$ and $b$, but it may decease the difference between the $x$-coordinates. Furthermore $\varrho$ is piecewise linear and continuous: It maps every line segment $ab\subset R'$ to a polygonal chain $\varrho(ab)$, and an $ab$-path $P'$ to an $\varrho(a)\varrho(b)$-path $P$. Since $\varrho$ is piecewise isometric, then $\|P\|=\|P'\|$.

By Lemma~\ref{lem:stairs}, there exists a geometric graph $G'$ comprised of $L'$ and additional edges of weight $O(\eps^{-1/2})$ such that $G'$ contains, for every $t'\in L'$, an $s't'$-path $P_{s't'}$ with $\|P_{s't'}\|\leq (1+O(\eps))\,|y(s')-y(t')|\leq (1+O(\eps))\,\|s't'\|$.

By construction, the graph $G'$ lies in $R'$. Let $G=\varrho(G')$, which is a geometric graph with possible new Steiner points at the vertical lines passing through the vertices of $L$. Since $\varrho$ is piecewise isometric, then $G$ is comprised of $L=\varrho(L')$ and additional edges of weight $O(\eps^{-1/2})$. For every $t\in L$, there exists a point $t'\in L'$ with $\varrho(t')=t$. As $G'$ contain an $s't'$-path $P_{s't'}$ with $\|P_{s't'}\|\leq (1+O(\eps))\,|y(s')-y(t')|$, then $G$ contains the $st$-path $P_{st}=\varrho(P_{s't'})$ with $\|P_{st}\|\leq \|P_{s't'}\|\leq  (1+O(\eps))\,|y(s')-y(t')| =(1+O(\eps))\,|y(s)-y(t)|
\leq (1+O(\eps))\,\|st\|$.
\end{proof}

\subsection{Combination of Shallow-Light Trees}
\label{ssec:combination}

We end this section with an easy corollary of Lemma~\ref{lem:monotone}, and show that
the combination of two \textsf{SLT}s yields a light $(1+\eps)$-spanner between points on two staircases.

\begin{lemma}\label{lem:combine2}
Let $\eps\in (0,1]$, let $R$ be an axis-parallel rectangle of width $1$ and height $2\eps^{-1/2}$; and let $L_1$ and $L_2$ be $y$-monotone paths lying in the vertical strip spanned by $R$ such that they each contain horizontal edges of total weight at most 1, the bottom vertex of $L_1$ is on the top side of $R$, and the top vertex of $L_2$ is on the bottom side of $R$;
see Fig.~\ref{fig:combinations}. Then there exists a geometric graph comprised of $L_1\cup L_2$ and additional edges of weight $O(\eps^{-1/2})$ that contains an $ab$-path $P_{ab}$ with $\|P_{ab}\|\leq (1+O(\eps))\,\|ab\|$ for any $a\in L_1$ and any $b\in L_2$.
\end{lemma}

\begin{figure}[htbp]
 \centering
 \includegraphics[width=0.65\textwidth]{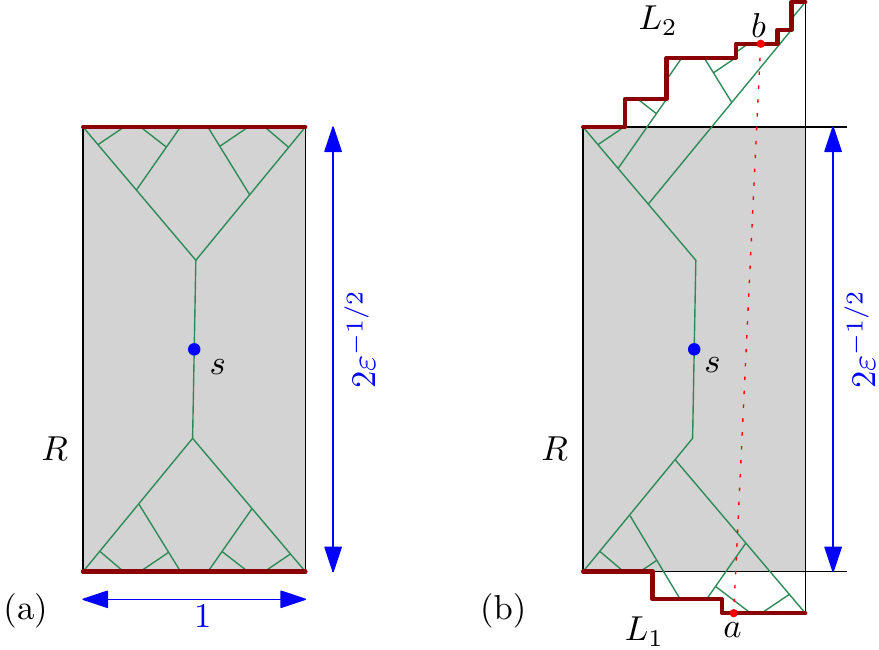}
 \caption{(a) A combination of two \textsf{SLT}s between the two horizontal sides of $R$.
 (b) A combination of two \textsf{SLT}s between two staircases above and below $R$, respectively.}
    \label{fig:combinations}
\end{figure}

\begin{proof}
Let $s$ be the center of the rectangle $R$. Let $G$ be the geometric graph formed by the \textsf{SLT}s from the source $s$ to $L_1$ and $L_2$, resp., using Lemma~\ref{lem:monotone}. By construction, $\|G\|=\|L_1\|+\|L_2\|+O(\eps^{-1/2})$.
It remains to show that $G$ has the desired spanning ratio.
Let $a\in L_1$ and $b\in L_2$. Let $h_a$ be the distance of $a$ from bottom side of $R$, and $h_b$ the distance of $b$ from the top side of $R$. By Lemma~\ref{lem:shallow}, the two \textsf{SLT}s jointly contain an $ab$-path $P_{ab}$ of length
$\|P_{ab}\|\leq (1+O(\eps))\,(\|as\|+\|bs\|)$.

On the one hand, $s$ is the center of $R$, and so $\|as\|+\|bs\|\leq \diam(R)+h_a+h_b\leq (1+\frac{\eps}{8})2{\eps}^{-1/2}+h_a+h_b$.
On the other hand, $\|ab\|\geq \mathrm{height}(R)+h_a+h_b=2{\eps}^{-1/2}+h_a+h_b$.
Overall, $\|P_{ab}\|\leq (1+O(\eps))(1+\frac{\eps}{8})\,\|ab\| \leq (1+O(\eps))\|ab\|$.
\end{proof}

\section{Construction of Directional Spanners for Staircases}
\label{sec:staircases}

In this section, we handle the special case of a finite set $S$ on a staircase path $L$. Our recursive construction uses special regions that we define now.
%
%
Let $L$ be an $x$- and $y$-monotone increasin staircase path, and let $P=P(L)$ be the staircase polygon bounded by $L$ above and left and by the boundary of the bounding box of $L$ from below and right. For $\lambda>0$, we define the \emph{$\lambda$-shadow of the vertical edges of $L$}, denoted by $\lambda$-$\shad_v(L)$, the set of points $p\in P$ such that there exists $a\in L$ on some vertical edge of $L$ such that $\mathrm{slope}(ap)\geq \lambda$; see Fig.~\ref{fig:shadow}(a).
Similarly we can define the $\lambda$-shadow of horizontal edges of $L$, $\lambda$-$\shad_h(L)$, be the set of points $p\in P$ such that there exists $b\in L$ on some horizontal edge of $L$ such that $\mathrm{slope}(bp)\geq \lambda$.
The region $\lambda$-$\shad_v(L)$ is not necessarily connected,
each connected components is bounded by a subpath of $L$ and a single line segment of slope $\lambda$.

\begin{figure}[htbp]
 \centering
 \includegraphics[width=0.9\textwidth]{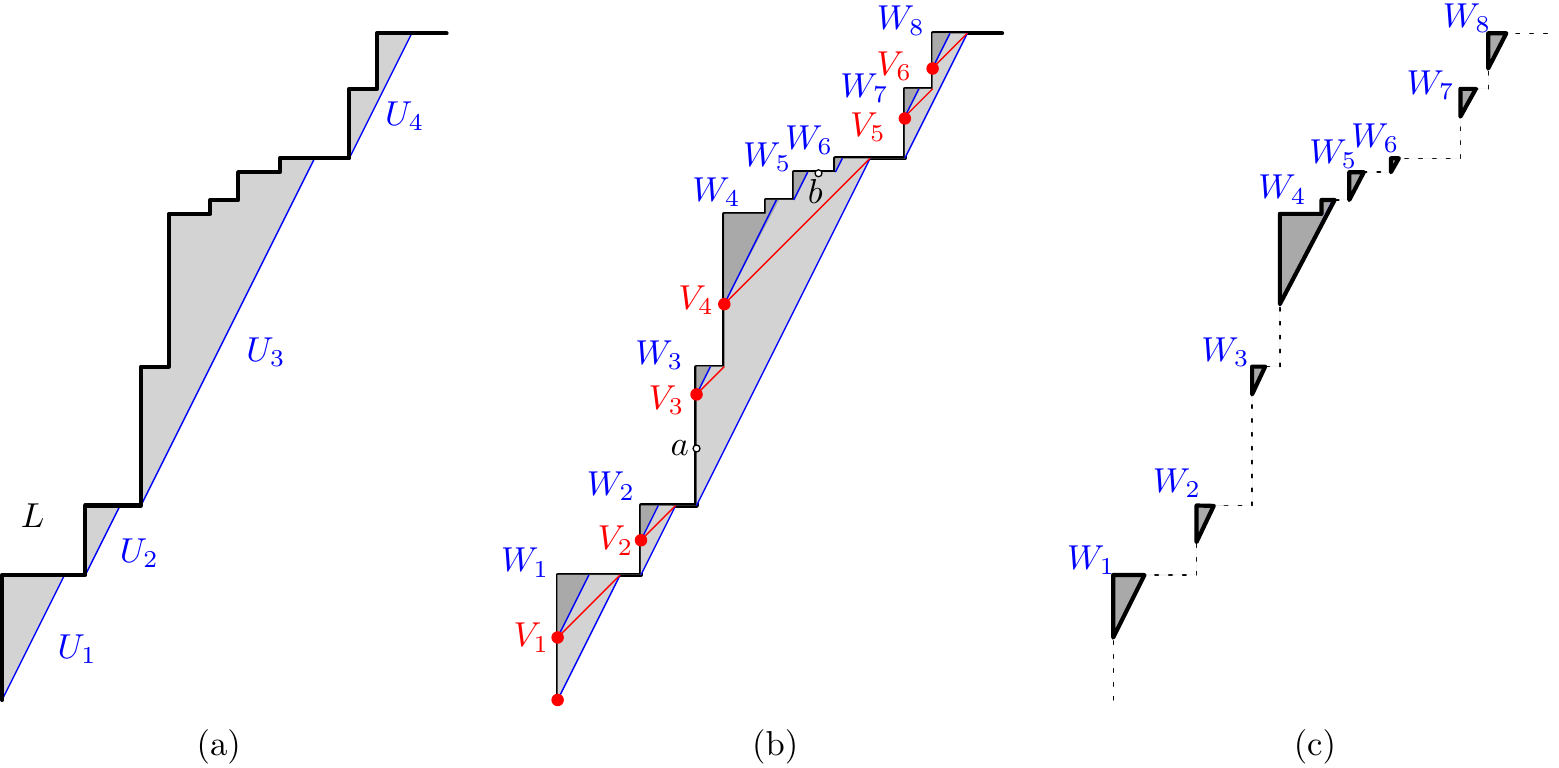}
 \caption{(a) A staircase path $L$; the shadow of vertical edges of $L$ is shaded light gray.
 (b) The shadow of the horizontal edges of the polygons $U_,\ldots, U_4$.
 (c) Recursive subproblems generated in the proof of Lemma~\ref{lem:staircase}.}
    \label{fig:shadow}
\end{figure}

\begin{lemma}\label{lem:staircase}
Let $L$ be a staircase path and $S\subset L$ a finite set.
Then there exists a geometric graph $G$ comprised of $L$ and additional edges of weight $O(\eps^{-1/2}\mathrm{width}(L))$ such that $G$ contains a path $P_{ab}$ of weight $\|P_{ab}\|\leq (1+O(\eps))\|ab\|$ for any $a,b\in S$ where $\slope(ab)\geq \eps^{-1/2}$ and the line segment $ab$ lies below $L$.
\end{lemma}
\begin{proof}
If $a,b\in L$ and $ab$ lies below $L$, then either both $a$ and $b$ are in the same edge of $L$ (hence $L$ contains a straight-line path $ab$), or one point in $\{a, b\}$ is on a vertical edge of $L$ and the other is on a horizontal edge of $L$. We may assume w.l.o.g.\ that $a$ is on a vertical edge and $b$ is on a horizontal edge of $L$.

Let $A=(\eps^{-1/2})$-$\shad_v(L)$ be the $(\eps^{-1/2})$-shadow of the vertical edges of $L$; see Fig.~\ref{fig:shadow}(a). Let $\mathcal{U}$ be the set of connected components of $A$. Note that for every pair $a,b\in L$, if $\mathrm{slope}(ab)\geq \eps^{-1/2}$ and $ab$ lies below the path $L$, then $ab$ lies in some polygon in $\mathcal{U}$.
For each polygon $U\in \mathcal{U}$, we construct a geometric graph $G(U)$ of weight $O(\eps^{-1/2}\wdth(U))$ such that $G(U)\cup L$ is a directional spanner for the point pairs in $S\cap U$. Then $L$ together with $\bigcup_{U\in \mathcal{U}} G(U)$ is a directional spanner for all possible $ab$ pairs. Since the polygons in $\mathcal{U}$
are adjacent to disjoint portions of $L$, we have $\sum_{U\in \mathcal{U}} \mathrm{width}(U)\leq \mathrm{width}(L)$, and so $\sum_{U\in \mathcal{U}}\|G(U)\|=O(\eps^{-1/2}\,\mathrm{width}(L))$, as required.

\smallskip\noindent\textbf{Recursive Construction.}
For every $U\in \mathcal{U}$, we construct $G(U)$ recursively as follows.
If $|S\cap U|\geq 2$, then let $B(U)=(\frac12\,\eps^{-1/2})$-$\shad_h(L\cap U)$ be the $(\frac12 \eps^{-1/2})$-shadow of the horizontal edges of $L\cap U$; see Fig.~\ref{fig:shadow}(b). Otherwise (if $|S\cap U|\leq 1$), let $B(U)=\emptyset$.
Denote by $\mathcal{V}$ the set of connected components of $B(U)$ for all $U\in \mathcal{U}$.

For every $V\in \mathcal{V}$, let $C(V)=(\eps^{-1/2})$-$\shad_v(L\cap V)$ be the $(\eps^{-1/2})$-shadow of the vertical edges of $L\cap V$; see Fig.~\ref{fig:shadow}(b).
Denote by $\mathcal{W}$ the set of all connected components of $C(V)$ for all $V\in \mathcal{V}$.

Since $\hght(W)/\wdth(W)=\eps^{-1/2}$ for all $W\in \mathcal{W}$ and
$\hght(V)/\wdth(V)=\frac12\eps^{-1/2}$ for all $V\in \mathcal{V}$, we have
\begin{align}
\sum_{W\in \mathcal{W}}\mathrm{width}(W)
&= \sqrt{\eps}\cdot \sum_{W\in \mathcal{W}}\mathrm{height}(W)
\leq \sqrt{\eps}\cdot \sum_{V\in \mathcal{V}}\mathrm{height}(V)\nonumber\\
&= \frac12\, \sum_{V\in \mathcal{V}}\mathrm{width}(V)
\leq \frac12\, \sum_{U\in \mathcal{U}}\mathrm{width}(U). \label{eq:width0}
\end{align}

For every polygon $V\in \mathcal{V}$, let $s_V$ be the bottom vertex of $V$.
We construct a sequence of shallow-light trees from source $s_V$ as follows.
For every nonnegative integer $i\geq 0$, let $h_i$ be a horizontal line at distance $\mathrm{height}(V)/2^i$ above $s_V$. If there is any point in $S$ between $h_i$ and $h_{i+1}$,
then we construct an \textsf{SLT} from $s_V$ to the portion of $L$ between $h_i$ and $h_{i+1}$.
By Lemma~\ref{lem:stairs}, the total weight of these trees is $O(\eps^{-1/2}\mathrm{width}(V))$.
Over all $V\in \mathcal{V}$, the weight of these \textsf{SLT}s is $\sum_{V\in \mathcal{V}} O(\eps^{-1/2}\mathrm{width}(V))
=O(\eps^{-1/2}\mathrm{width}(U))$.
For all $V\in \mathcal{V}$, we also add the boundary $\partial V$ to our spanner,
at a cost of $\sum_{V\in \mathcal{V}} \per(V) =\sum_{V\in \mathcal{V}} O(\eps^{-1/2}\mathrm{width}(V))=O(\eps^{-1/2}\mathrm{width}(U))$.
This completes the description of one iteration.
Recurse on all $W\in \mathcal{W}$ that contain any point in $S$.

\smallskip\noindent\emph{Lightness analysis.}
Each iteration of the algorithm, for every polygon $U\in \mathcal{U}$, constructs \textsf{SLT}s of total weight $O(\eps^{-1/2}\wdth(U))$, and produces subproblems whose combined width is at most $\frac12\wdth(U)$ by Equation~\eqref{eq:width0}.
Consequently, summation over all levels of the recursion yields
$\|G(U)\|=O(\eps^{-1/2}\wdth(U) \cdot \sum_{i\geq 0}2^{-i})=O(\eps^{-1/2}\wdth(U))$, as required.

\smallskip\noindent\emph{Stretch analysis.}
Now consider a point pair $a,b\in S$ such that $\slope(ab)\geq \eps^{-1/2}$,
$a$ is in a vertical edge of $L$, and $b$ is in a horizontal edge of $L$. Assume that $U$ is the smallest shadow polygon in the recursion above that contains both $a$ and $b$. Then $b\in V$ for some $V\in \mathcal{V}$, and $a$ is at or below vertex $s_V$ of $V$. Now we can find an $ab$-path $P_{ab}$ as follows:
First construct a $y$-monotonically increasing path from $a$ to $s_V$ along vertical edges of $L$ and along edges of slope $\frac{1}{2}\eps^{-1/2}$ of polygons in $\mathcal{V}$. 
Then an \textsf{SLT} contains a path from $s_V$ to $b$.
All edges of $P_{ab}$ from $a$ to $s_V$ are vertical or have slope $\frac12\eps^{-1/2}$, and so their directions differ from vertical by at most $\mathrm{arctan}(2\eps^{1/2})\leq 3\eps^{1/2}$, using the Taylor expansion of $\tan(x)$ near $0$.
By Lemma~\ref{lem:angle2} the stretch factor of the $as_V$-path and the path $(a,s_V,b)$ are each at most $1+O(\eps)$.
By Lemma~\ref{lem:staircase}, the \textsf{SLT} contains a $s_Vb$-path of stretch factor $1+O(\eps)$. Overall, $\|P_{ab}\|\leq (1+O(\eps))\|ab\|$.
\end{proof}

In Section~\ref{sec:tame}, we show that Lemma~\ref{lem:staircase} continues to hold if we replace the horizontal edges with $x$-monotone tame paths. Specifically, Lemma~\ref{lem:tame-staircase} below generalizes this result.

\section{Directional Spanners for Tame Histograms}
\label{sec:tame}

In this section we prove Lemma~\ref{lem:hist} for tame histograms.
Given a tame histogram $T$ and a finite set of points $S\subset \partial T$, we construct a directional spanner for $S$ with respect to the interval $D=[\frac{\pi-\sqrt{\eps}}{2},\frac{\pi+\sqrt{\eps}}{2}]$ of nearly-vertical directions. In the discussion below, we typically use $\slope(ab)$, rather than $\dir(ab)$. Note that whenever $a,b\in S$ and $\dir(ab)\in D$, then $|\slope(ab)|\geq \eps^{-1/2}$, due to the Taylor estimate $\tan(x)\geq x+x^3/3$ for $x=\dir(ab)-\frac{\pi}{2}$ near 0.

The next lemma (Lemma~\ref{lem:fat}) establishes a key property of tame histograms: the weight of a subpath between two points in the $pq$-path can be bounded in terms of the $L_1$-distance between the two endpoints and an error term dominated by the distance between their $x$-coordinates.

\begin{lemma}\label{lem:staircase-path}
Let $H$ be an $x$-monotone histogram bounded by a horizontal segment $pq$ and a $pq$-path $L$. Let $a,b\in L$ such that $b$ is the bottom-most point in $L_{ab}$. Then there exists a staircase $ab$-path $P_{ab}$ comprised of segments in $L_{ab}$ and horizontal chords of $L_{ab}$.
\end{lemma}

\begin{figure}[htbp]
 \centering
 \includegraphics[width=0.45\textwidth]{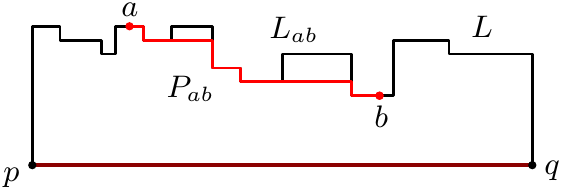}
 \caption{A tame histogram and the $ab$-path $P{ab}$ constructed in the proof of Lemma~\ref{lem:fat}.}
    \label{fig:descending}
\end{figure}

\begin{proof}
Assume w.l.o.g.\ that $x(a)\leq x(b)$; refer to Fig.~\ref{fig:descending}. We construct an $ab$-path $P$ incrementally as follows: Initially, we set $P=(a)$ to be the one-vertex path, and then incrementally append new edges until it reaches $b$. Let $c$ denote the current endpoint of $P$. 

While $c\neq b$ do as follows. If $c$ is in a vertical edge $e$ of $L_{ab}$, but not the bottom endpoint of $e$, then we extend $P$ to the bottom endpoint of $e$. Else if $c$ is in a horizontal edge $e$ of $L_{ab}$, but not the right endpoint of $e$, then we extend $P$ to the right endpoint of $e$. Else $c$ is the bottom point of a vertical edge $e_v$ and the right endpoint of a horizontal edge $e_h$ of $L_{ab}$, and then we extend $P$ with a horizontal chord $cd$. Such a chord exists since $b$ is the bottommost point in $L_{ab}$. The algorithm terminates with $c=b$, since in each iteration either $y(c)$ decreases, or $y(c)$ does not change but $x(c)$ increases.
\end{proof}

\begin{lemma}\label{lem:fat}
Let $H$ be a tame histogram bounded by a horizontal segment $pq$ and a $pq$-path $L$.
Let $a,b\in L$ such that $b$ is the bottom-most point in $L_{ab}$.
Then $\|L_{ab}\|\leq 2|x(a)-x(b)|+|y(a)-y(b)|$.
\end{lemma}

\begin{proof}
By Lemma~\ref{lem:staircase-path}, there exists a staircase $ab$-path $P_{ab}$ that comprises portions of $L_{ab}$ and horizontal chords of $L_{ab}$. Since $P_{ab}$ is a staircase $ab$-path, the total weight of
its vertical edges is $|y(a)-y(b)|$ and the total weight of its horizontal edges is $|x(a)-x(b)|$.
If we replace every horizontal chord $cd$ along $P_{ab}$ with the corresponding subpath $L_{cd}$ of $L$, the resulting path is precisely $L_{ab}$. As $H$ is a tame histogram, each chord $cd$ is replaced by a path of length at most $2\|cd\|=2\, |x(c)-x(d)|$.
As these chords are disjoint horizontal line segments along $P_{ab}$,
the weight increase is bounded by
$\|L_{ab}\|-\|P_{ab}\|\leq |x(a)-x(b)|$.
Consequently,
\[
\|L_{ab}\|
=(\|L_{ab}\|-\|P_{ab}\|)+\|P_{ab}\|
\leq 2|x(a)-x(b)|+|y(a)-y(b)|,
\]
as claimed.
\end{proof}

In Lemmas~\ref{lem:combine5} and~\ref{lem:combine6} below,
we use \textsf{SLT}s to construct directional $(1+\eps)$-spanners in a tame histogram
(i) between the base and a portion of the path $L$ within a square; and
(ii) between a source $s$ and a portion of the path $L$ from $p$ to $q$.

\begin{figure}[htbp]
 \centering
 \includegraphics[width=0.85\textwidth]{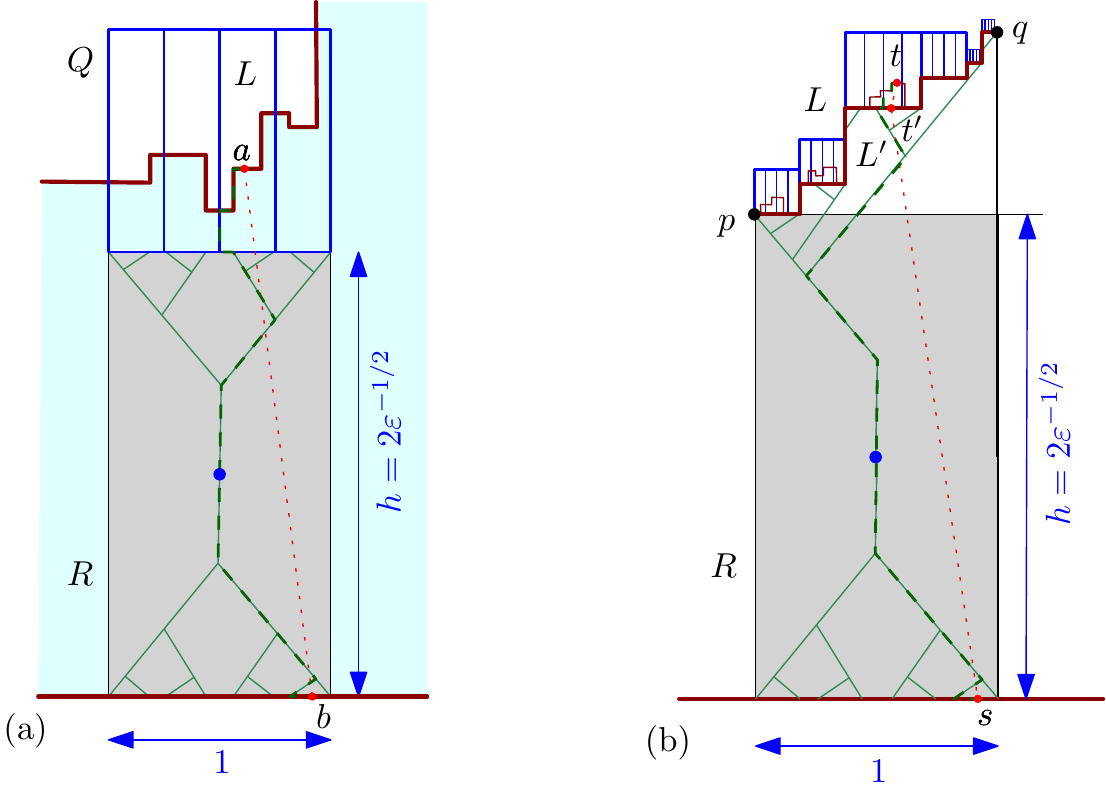}
 \caption{(a) The boundary of a tame histogram in a square $Q$ above rectangle $R$.
 (b) An adaptation of an \textsf{SLT} to tame histograms.}
    \label{fig:combinations2}
\end{figure}

\begin{lemma}\label{lem:combine5}
Let $R$ be an axis-parallel rectangle of width $1$ and height $2\eps^{-1/2}$.
Let $Q$ be a unit square adjacent to the top side of $R$, and let $L$ be a tame path in $Q$; see Fig.~\ref{fig:combinations2}(a).
Then there exists a graph $G$ comprised of $L$ and additional edges of weight $O(\eps^{-1/2})$ that contains
an $ab$-path $P_{ab}$ with $\|P_{ab}\|\leq (1+\eps)\,\|ab\|$ for any $a\in L$ and any point $b$ in the bottom side of $R$.
\end{lemma}
\begin{proof}
Let $s$ be the center of the rectangle $R$. We construct a geometric graph $G$ as follows. Let $G$ contain the bottom side of $R$, the path $L$, and the two \textsf{SLT}s from $s$ to the bottom and top sides of $R$, respectively. Let $G$ also contain a subdivision of $Q$ into rectangles of aspect ratio $\eps^{-1/2}$; see Fig.\ref{fig:combinations2}(a). Specifically, we subdivide $Q$ into rectangles
$r$ of $\wdth(r)=\sqrt{\eps}$ and $\hght(r)=1$. Finally, in each rectangle $r\subset Q$ of this subdivision, if $r$ intersects $L$, then let $G$ contain a vertical line segment from a bottom-most point in $L\cap r$ to the bottom side of $Q$.

\smallskip\noindent\emph{Lightness analysis.}
The weight of two \textsf{SLT}s is $O(\eps^{-1/2})$ by Lemma~\ref{lem:shallow}. Since $\wdth(Q)=1$, the weight of the subdivision of $Q$ is $O(\eps^{-1/2})$, and the weight of the vertical edge in each rectangle $r\subset Q$ is at most $\hght(r)=\hght(Q)=1$. The overall weight of $G$ is $\|L\|+O(\eps^{-1/2})$.

\smallskip\noindent\emph{Stretch-factor analysis.}
Let $a\in L$ and let $b$ be a point in the bottom side of $R$. We may assume that $a\in r$, for a rectangle $r\subset Q$ in the subdivision of $Q$. We construct an $ab$-path $P_{ab}$ as follows: Start from $a$, follow $L$ to a bottom-most point in $L\cap r$, and then use a vertical line segment to reach the bottom side of $r$. Then use the two \textsf{SLT}s to reach $b$. For easy reference, we label some of intermediate vertices along $P_{ab}$: let $v_1$ be the bottom-most point in $L\cap r$, and let $v_2$ be the bottom endpoint of the vertical segment in $r$, where $P_{ab}$ reaches the top side of $R$. Note that the $y$-coordinates of these points monotonically decrease, that is, $y(a)\geq y(v_1)\geq y(v_2)\geq y(b)$. Clearly, we have $\|ab\|\geq y(a)-y(b)\geq y(v_2)-y(b)= 2\eps^{-1/2}$.

We now estimate the weight of each portion of $P_{ab}$ between $a$, $v_1$, $v_2$, and $b$. By Lemma~\ref{lem:fat},
we have
\begin{align*}
\|P_{av_1}\|
&\leq 2\, |x(a)-x(v_1)| + |y(a)-y(v_1)|\\
&\leq 2\,\wdth(r)+|y(a)-y(v_1)|\\
&\leq 2\,\sqrt{\eps}+(y(a)-y(v_1)).
\end{align*}
As $v_1v_2$ is a vertical line segment, then $\|P_{v_1 v_2}\| = y(v_1)-y(v_2)$.
Lemma~\ref{lem:combine2} yields $\|P_{v_2 b}\|\leq (1+O(\eps)) (y(v_2)-y(b))$.
Putting the pieces together, we obtain
\begin{align*}
\|P_{ab}\|
&=\|P_{a v_1}\| + \|P_{v_1 v_2}\| +  \|P_{v_3 b}\|\\
&\leq 2\,\sqrt{\eps}+(1+O(\eps))
    \Big((y(a)-y(v_1)) + (y(v_1)-y(v_2))+ (y(v_2)-y(b))\Big)\\
&\leq (1+O(\eps))(y(a)-y(b)) + 2\,\sqrt{\eps}\\
&\leq (1+O(\eps))\|ab\|,
\end{align*}
as required.
\end{proof}

\begin{lemma}\label{lem:combine6}
Let $R$ be an axis-parallel rectangle of width 1 and height $2\eps^{-1/2}$; and let $p$ be the upper-left corner of $R$, and let $q$ be a point on vertical line passing through the right side of $R$ with $y(p)\leq y(q)$.
Let $L$ be a tame $pq$-path that lies above the line segment $pq$; see Fig.~\ref{fig:combinations2}(b).
Then there exists a geometric graph $G$ comprised of $L$ and additional edges of weight $O(\eps^{-1/2})$ that contains an $st$-path $P_{st}$ with $\|P_{st}\|\leq (1+O(\eps))\,\|st\|$ for any $s$ in the bottom side of $R$ and any $t\in L$.
\end{lemma}
\begin{proof}
By Lemma~\ref{lem:staircase-path}, there exists a staircase $qp$-path $P_{qp}$ comprised of segments of $L$ and horizontal chords of $L$. 
Traversing $P_{qp}$ from $p$ to $q$, we obtain a $x$- and $y$-monotone increasing $pq$-path that we denote by $L'$. For each horizontal edge $e$ of $P_{pq}$, let $Q_e$ be an axis-parallel square of side length $\|e\|$ above $e$. Since $L$ is a tame path, each connected component of $L\setminus L'$ lies in a square $Q_e$ for some horizontal edge $e$ of $L'$; see Fig.~\ref{fig:combinations2}(b).

We construct a geometric graph $G$ as follows. Let $G$ contain two \textsf{SLT}s from the center of $R$ to the bottom side of $R$ and to $L'$, resp., described in Lemma~\ref{lem:stairs}.
Let $G$ also contain a subdivision of each square $Q_e$ into rectangles of aspect ratio $2\eps^{-1/2}$. Finally, in each rectangle $r\subset Q$ of this subdivision, if $r$ intersects $L$, then $G$ contains a vertical line segment from a bottom-most point in $L\cap r$ to the bottom side of $r$. The weight of the \textsf{SLT} is $O(\eps^{-1/2})$ by Lemma~\ref{lem:stairs}. Since the sum of the widths of all squares $Q_e$ is at most one 1, the total weight of the grids in $Q_e$ is also $O(\eps^{-1/2})$, and the vertical edges in the rectangles in $r\subset Q$ are bounded above by the weight of the grid. The overall weight of $G$ is $\|L\|+O(\eps^{-1/2})$.

Let $S$ be a point in the bottom side of $R$, and $t\in L$. If $t\in L'$, then the two \textsf{SLT}s jointly contain a path $P_{st}$ with $\|P_{st}\|\leq (1+O(\eps))\|st\|$ by Lemma~\ref{lem:combine2}. Otherwise, $t\in L\setminus L'$. Since $L$ is a tame path, $t$ lies in a square $Q_e$ for some horizontal edge $e$ of $L'$. We can construct a path $P_{st}$ as a path from $t$ to a point $t'\in e$ similarly to the proof of Lemma~\ref{lem:combine5}, followed by a path from $t'$ to $s$ in the \textsf{SLT}s.
\end{proof}

We use Lemma~\ref{lem:combine5} to construct a $(1+\eps)$-spanner between the base $pq$ and $pq$-path in a tame histogram.

\begin{lemma}\label{lem:dir2}
Let $H$ be a tame histogram bounded by a horizontal line $pq$ and $pq$-path $L$,
and let $S\subset \partial H$ be a finite point set.
Then there exists a geometric graph $G$ of weight $\|G\|=O(\eps^{-1/2}\,\per(P))$
such that $G$ contains a $ab$-path $P_{ab}$ with $\|P_{ab}\|\leq (1+\eps)\|ab\|$
for all $a\in S\cap L$ and $b\in pq$ such that $ab\subset H$ and $|\slope(ab)|\geq \eps^{-1/2}$.
\end{lemma}
\begin{proof}
We construct a collection $\mathcal{Q}$ of squares such that
for every square $Q\in \mathcal{Q}$ is adjacent to a rectangle $R(Q)$ as in
the setting of Lemma~\ref{lem:combine5}; and for every point pair $a,b\in S$,
with $a\in L$ and $b\in pq$, there is a square $Q\in \mathcal{Q}$ such that $a,b\in Q\cup R(Q)$. Let $G(Q)$ be the geometric graph in Lemma~\ref{lem:combine5} for all $Q\in \mathcal{Q}$, and let $G=\bigcup_{Q\in \mathcal{Q}}G(Q)$. Then $G$ has the required stretch factor.
It remains to construct the collection $\mathcal{Q}$ of squares, and
show that $\|G\|=O(\eps^{-1/2}\,\per(P))$.

\smallskip\noindent\emph{Construction of Squares.}
Refer to Fig.~\ref{fig:thames}.
Let $H$ be a tame histogram bounded by a horizontal line $pq$ and $pq$-path $L$.
We may assume w.l.o.g. that $p$ is the origin and $q$ is on the positive $x$-axis,
and $h=\mathrm{height}(H)$.
Since $H$ is tame, $\|L\|\leq 2\|pq\|$, which implies that $h\leq \frac12\,\|pq\|$.
For every nonnegative integer $i\in \mathbb{N}$, let
\[\ell_i:y=h \left(1-3\cdot \sqrt{\eps}\right)^i.\]
We tile the horizontal strip between two consecutive lines, $\ell_i$ and $\ell_{i+1}$,
by squares in two different ways, such that the midpoint of a square in one tiling is on the
boundary of two squares in the other tiling.

\begin{figure}[htbp]
 \centering
 \includegraphics[width=0.8\textwidth]{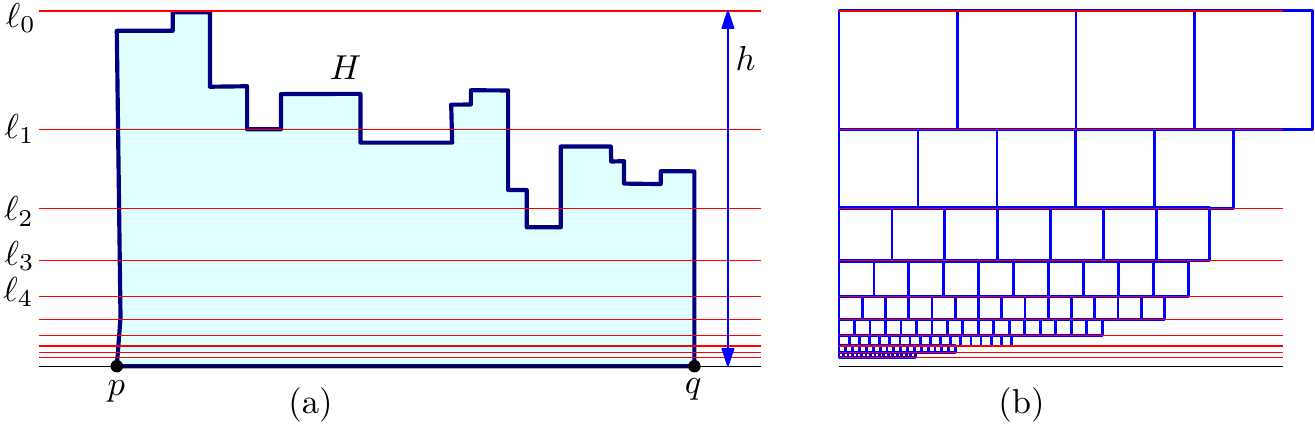}
 \caption{(a) A tame histogram $H$ and horizontal lines $\ell_i$, $i\in \mathbb{N}$, defined in the proof of Lemma~\ref{lem:dir2}.
 (b) Tiling of the horizontal strips between consecutive lines $\ell_i$ and $\ell_{i+1}$.}
    \label{fig:thames}
\end{figure}

Let $\mathcal{Q}$ be the set of squares $Q$ in the tilings defined above such that $Q\cap S\neq \emptyset$. For each square $Q\in \mathcal{Q}$, let $R(Q)$ be the axis-aligned rectangle such that the top side of $R(Q)$ equals the bottom side of $Q$, and the bottom side of $R(Q)$ is in the $x$-axis. Then the aspect ratio of $R(Q)$ is $\frac13\,\eps^{-1/2}$. Indeed, if $Q$ lies between $\ell_i$ and $\ell_{i+1}$, then 
\[
\frac{\hght(R(Q))}{\wdth(R(Q))}
=\frac{\hght(R(Q))}{\hght(Q)}
=\frac{h\left(1-3 \sqrt{\eps}\right)^i}{ h\left(1-3 \sqrt{\eps}\right)^i - h\left(1-3 \sqrt{\eps}\right)^{i+1}} 
=\frac{1}{1- \left(1-3 \sqrt{\eps}\right) }
=\frac{1}{3 \sqrt{\eps}}.
\]
For every $Q\in \mathcal{Q}$, Lemma~\ref{lem:combine5} (invoked with $36\,\eps$ in place of $\eps$) yields a geometric graph $G(Q)$. Let $G=\bigcup_{Q\in \mathcal{Q}}G(Q)$.

\smallskip\noindent\emph{Lightness Analysis.}
By Lemma~\ref{lem:combine5}, the graph $G(Q)$ is comprised of $L\cap Q$ and additional edges of weight $O(\eps^{-1/2}\wdth(Q))$. For the desired bound $\|G\|\leq O(\eps^{-1/2}\per(H))$,
it is enough to prove that $\sum_{Q\in \mathcal{Q}} \wdth(Q) \leq O(\per(H))$.

We define a proximity graph $\widehat{G}$ on the squares in $\mathcal{Q}$.
The vertex set is $V(\widehat{G})=\mathcal{Q}$, and squares $Q_1,Q_2\in \mathcal{Q}$
are adjacent if and only if $\mathrm{dist}(Q_1,Q_2)\leq \wdth(Q_1)+\wdth(Q_2)$.
Since the squares in the horizontal strip between $\ell_i$ and $\ell_{i+1}$ form two tilings,
and the widths of the squares in adjacent horizontal strips differ by a factor close to 1,
the maximum degree in $\widehat{G}$ is $O(1)$. Consequently, $\widehat{G}$ is
$O(1)$-degenerate, and we can partition its vertex set $\mathcal{Q}$ into
$O(1)$ independent sets.

For every $Q\in \mathcal{Q}$, let $2Q$ denote the square obtained by dilating $Q$ from its center by a factor of 2.
Since $L$ contains points in $Q$, but at least one of its endpoints is outside of $2Q$, then $L$ traverses the annulus $2Q\setminus Q$, which implies $\|L\cap 2Q\|\geq \mathrm{width}(Q)$. For an independent set $\mathcal{I}\subset \mathcal{Q}$, the squares $\{2Q: Q\in \mathcal{I}\}$ are pairwise disjoint.
It follows that
\begin{equation}\label{eq:ind}
\sum_{Q\in \mathcal{I}} \mathrm{width}(Q)
\leq \sum_{Q\in \mathcal{I}}\|L\cap 2Q\|
\leq \left \| L\cap \left(\bigcup_{Q\in \mathcal{I}} 2Q\right)\right\|
\leq \|L\|.\nonumber
\end{equation}
Summation over $O(1)$ independent sets yields  $\sum_{Q\in \mathcal{Q}} \mathrm{width}(Q) \leq \|L\|\leq O(\per(H))$,
as required.

\smallskip\noindent\emph{Stretch analysis.}
Now consider points $a\in S\cap L$ and $b\in pq$ such that $ab\subset H$ and $|\slope(ab)|\geq \eps^{-1/2}$.
Assume w.l.o.g.\ that $\slope(ab)>0$. There exists a square $Q\in \mathcal{Q}$ such that $a$ lies in the right half of $Q$.
We have $x(b)-x(a)\leq (y(a)-y(b))/\slope(ab)\leq\eps^{1/2} \hght(Q\cup R(Q)) \leq \frac32\eps^{1/2} \hght(R(Q))\leq \frac12 \, \wdth(R(Q))$. Consequently, $b$ is on the bottom side of $R(Q)$, and so $G(Q)$ contains an $ab$-path of weight $(1+O(\eps))\|ab\|$ by Lemma~\ref{lem:combine5}. 
\end{proof}

In the remainder of this section, we construct a directional $(1+\eps)$-spanner for points on the $x$-monotone path of a tame histogram. This is done by an adaptation of Lemma~\ref{lem:staircase}. Even though the horizontal edges are replaced by tame paths, the weight analysis remains essentially the same.

The crucial observation in the proof of Lemma~\ref{lem:staircase}
(cf.~Equation~\eqref{eq:width0})
was that if $L$ is an $x$- and $y$-monotone increasing staircase $ab$-path,
then $\mathrm{slope}(ab)=\mathrm{height}(L)/\mathrm{width}(L)$.
We show that this equation holds approximately for any tame path $L$,
where the width and height of $L$ are replaced by the total weight
of horizontal and vertical edges of $L$, denoted $\hper(L)$ and $\vper(L)$, respectively.

\begin{lemma}\label{lem:width}
If $L$ is a tame $pq$-path such that $\frac12\eps^{-1/2}\leq \slope(pq)\leq \eps^{-1/2}$ for $\eps\in (0,\frac{1}{16}]$, then
\begin{equation}\label{eq:width2}
\slope(pq) \leq \frac{\vper(L)}{\hper(L)} \leq \frac{3}{2}\, \slope(pq).
\end{equation}
\end{lemma}
\begin{proof}
Assume w.l.o.g. that $x(p)<x(q)$ and $y(p)<y(q)$.
By Lemma~\ref{lem:staircase-path}, there exists a staircase $qp$-path $P_{qp}$ comprised of segments of $L$ and horizontal chords of $L$. 
Traversing $P_{qp}$ from $p$ to $q$, we obtain a $x$- and $y$-monotone increasing $pq$-path that we denote by $L'$;
For each horizontal chord $ab$ of $L$, we have $\hper(L_{ab})=\|ab\|$ since $L$ is $x$-monotone, and
$\vper(L_{ab})\leq \hper(ab)\leq \|ab\|$ since $L$ is tame.
Consequently, we have $\hper(L')=\hper(L)=\wdth(pq)$ and,
as $L'$ is a staircase path, then
 \begin{align*}
\vper(L)  \leq & \vper(L')+\hper(L')
    = \hght(pq)+\wdth(pq)\\
    \leq & \left(1+\frac{1}{\slope(pq)}\right)\hght(pq)
    \leq (1+2\eps^{1/2})\hght(pq).
 \end{align*}
Overall, we obtain
    \[ \slope(pq)=
     \frac{\hght(L')}{\wdth(L')}
     \leq \frac{\vper(L)}{\hper(L)}
     \leq  \frac{(1+2\eps^{1/2})\hght(L')}{\wdth(L')}
      \leq \frac32\,\slope(pq),
    \]
where we used that $0<\eps\leq \frac{1}{16}$ implies $1+2\eps^{1/2}\leq \frac32$.
\end{proof}

\begin{figure}[htbp]
 \centering
 \includegraphics[width=0.9\textwidth]{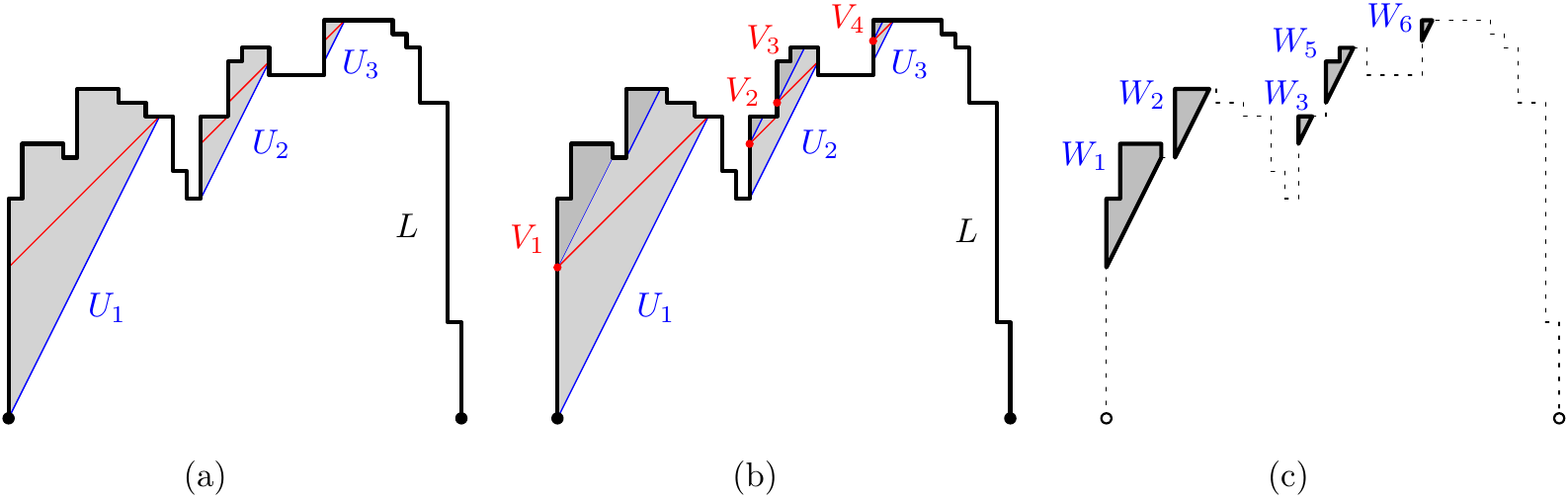}
 \caption{(a) A tame path $L$. The shadow of the ascending vertical edges of $L$ is shaded light gray.
 (b) The shadow of the horizontal edges and descending vertical edges is shaded dark gray.
 (c) Recursive subproblems generated in the proof of Lemma~\ref{lem:tame-staircase}.}
    \label{fig:tame-shadow}
\end{figure}

As noted above, the following lemma is an adaptation of Lemma~\ref{lem:staircase} to tame paths. Due to Lemma~\ref{lem:width}, the  recursive weight analysis carries over to this case. For clarity, we present a complete proof.

\begin{lemma}\label{lem:tame-staircase}
Let $L$ be a tame path, $S\subset L$ a finite point set, and $\eps\in (0,1]$.
Then there exists a geometric graph $G$ comprised of $L$ and additional edges of weight $O(\eps^{-1/2}\hper(L))$
such that $G$ contains a path $P_{ab}$ of weight $\|P_{ab}\|\leq (1+O(\eps))\|ab\|$ for any $a,b\in L$ such that $|\slope(ab)|\geq \eps^{-1/2}$ and the line segment $ab$ lies below $L$.
\end{lemma}
\begin{proof}
By Lemma~\ref{lem:fat}, we have $\|L_{ab}\|\leq 2|x(a)-x(b)|+|y(a)-y(b)|\leq (1+\sqrt{2})\|ab\|$,
so the claim holds for $G=L$ if $\eps\in (\frac{1}{16},1]$.
In the remainder of the proof, assume $\eps\in (0,\frac{1}{16}]$.
We construct $G$ as a union of two graphs, $G^+$ and $G^-$, where $G^+$ is a spanner for $\{a,b\}$ pairs with $\mathrm{slope}(ab)>0$ and $G^-$ for $\mathrm{slope}(ab)<0$. We focus on $G^+$, as the case of $G^-$ is analogous.

Let $a, b\in S$ such that $\mathrm{slope}(ab)\geq \eps^{-1/2}$ and $ab$ lies below $L$.
Without loss of generality, we may assume $y(a)<y(b)$. Since $ab$ is below $L$, point $a$ cannot be an interior of a horizontal edge of $L$. Note that $a$ is a point in an ascending vertical edge of $L$.

Let $A=(\eps^{-1/2})$-$\shad_v(L)$ be the $(\eps^{-1/2})$-shadow of vertical edges of $L$; see Fig.~\ref{fig:tame-shadow}(a). Let $\mathcal{U}$ be the set of connected  components of $A$. By construction every pair $a,b\in L$ with $\mathrm{slope}(ab)\geq \eps^{-1/2}$ and $ab\subset H$ lies in some polygon in $\mathcal{U}$.
For each polygon $U\in \mathcal{U}$, we construct a geometric graph $G^+(U)$ of weight $O(\eps^{-1/2}\hper(U))$
such that $G^+(U)\cup L$ is a directional $(1+\eps)$-spanner for the points in $S\cap U$. Then $L$ together with $\bigcup_{U\in \mathcal{U}} G^+(U)$ is $(1+\eps)$-spanner for all possible $ab$ pairs. Since the polygons in $\mathcal{U}$ are adjacent to disjoint portions of $L$, we have $\sum_{U\in \mathcal{U}} \hper(U)\leq \hper(L)$,
and so $\sum_{U\in \mathcal{U}}\|G^+(U)\|=O(\eps^{-1/2}\hper(L))$, as required.

\smallskip\noindent\textbf{Recursive Construction.}
For each $U\in \mathcal{U}$, we construct $G^+(U)$ recursively as follows.
If $|S\cap U|\geq 2$, then let $B(U)=(\frac12 \eps^{-1/2})$-$\shad_h(L\cap U)$ be the $(\frac12 \eps^{-1/2})$-shadow of horizontal edges of $L\cap U$; see Fig.~\ref{fig:tame-shadow}(b). $|S\cap U|\geq 2$, then let $B(U)=\emptyset$.
Denote by $\mathcal{V}$ the set of connected components of $B(U)$ for all $U\in \mathcal{U}$.

For every $V\in \mathcal{V}$, let $C(V)=(\eps^{-1/2})$-$\shad_v(L\cap V)$ be the $(\eps^{-1/2})$-shadow of vertical edges of $L\cap V$; see Fig.~\ref{fig:tame-shadow}(b). Denote by $\mathcal{W}$ the set of all connected components of $C(V)$ for all $V\in \mathcal{V}$.

We can apply Lemma~\ref{lem:width} with slope $\eps^{-1/2}$ for all $W\in \mathcal{W}$;
and  with slope $\frac12\,\eps^{-1/2}$ for all $V\in \mathcal{V}$. Then
\begin{align}
\sum_{W\in \mathcal{W}}\hper(W)
&\leq \frac32\cdot \sqrt{\eps}\cdot \sum_{W\in \mathcal{W}}\vper(W)
\leq \frac32\cdot \sqrt{\eps}\cdot \sum_{V\in \mathcal{V}}\vper(V)\nonumber\\
&\leq \frac32\cdot \frac12\, \sum_{V\in \mathcal{V}}\hper(V)
\leq \frac34\, \sum_{U\in \mathcal{U}}\hper(U).\label{eq:width1}
\end{align}
Consequently, $\sum_{W\in \mathcal{W}} \|G^+(W)\|$ is proportional to
$\frac34\cdot \eps^{-1/2}\sum_{U\in \mathcal{U}}\hper(U)$.

For each polygon $V\in \mathcal{V}$, let $s_V$ be the bottom vertex of $V$, let $L(V)=L\cap V$ be the portion of $L$ on the boundary of $V$. We construct a sequence of \textsf{SLT}s from source $s_V$ as follows. For every nonnegative integer $i\geq 0$, let $h_i$ be a horizontal line at distance $\mathrm{height}(V)/2^i$ above $s_V$. Let $L_i\subset L(V)$ be a maximum portion of $L(V)$ such that the corresponding staircase path $L'_i$ is on or below $h_i$ and strictly above $h_{i+1}$.
By Lemma~\ref{lem:combine6}, we can construct
a \textsf{SLT} from $s_V$ to $L_i$.
The total weight of these \textsf{SLT}s is $O(\eps^{-1/2}\hper(V))$.
Then the overall weight of these spanners is
$\sum_{V\in \mathcal{V}} O(\eps^{-1/2}\hper(V)) =O(\eps^{-1/2}\hper(U))$.
This completes the description of one iteration.
Recurse on all $W\in \mathcal{W}$ that contain any point in $S$.

\smallskip\noindent\emph{Lightness analysis.}
Each recursive call of the algorithm, for a polygon $U$, adds edges of total weight $O(\eps^{-1/2}\hper(U))$ to $G^+(U)$ and produces subproblems whose combined horizontal perimeter is at most $\frac34\hper(U)$ by Equation~\eqref{eq:width1}.
Consequently, summation over all subsequent levels of the recursion yields
$\|G^+(U)\|=
O(\eps^{-1/2}\hper(U) \cdot
\sum_{i\geq 0}\left(\frac34\right)^{-i})=O(\eps^{-1/2}\hper(U))$, as required.

\smallskip\noindent\emph{Stretch analysis.}
Now consider point pair $a,b\in S$ such that $\mathrm{slope}(ab)\geq \eps^{-1/2}$,
$a$ is in an ascending vertical edge of $L$, and $b$ is in a horizontal edge or a descending vertical edge of $L$. Assume that $U$ is the smallest polygon in the recursion above that contains both $a$ and $b$. Then $b\in V$ for some $V\in \mathcal{V}$, and $a$ is at or below the bottom vertex $s_V$ of $V$. Now we can find an $ab$-path $P_{ab}$ as follows:
First construct a $y$-monotonically increasing path from $a$ to $s_V$ along vertical edges of $L$ and along the edges of slope $\frac{1}{2}\eps^{1/2}$ of some polygons in $\mathcal{V}$. Then from $s_V$ to $b$, follow an \textsf{SLT} provided by Lemma~\ref{lem:combine6}. Specifically, let $b'$ be the orthogonal projection of $b$ to the staircase path $L'$. There exists an integer $i\geq 0$ such that $b'$ lies between the horizontal lines $h_i$ and $h_{i+1}$, and we can use the \textsf{SLT} constructed between $s_V$ and $L_i$.

All edges of $P_{ab}$ from $a$ to $s_V$ are vertical or have slope $\frac12\eps^{-1/2}$, and so their directions differ from vertical by at most $\mathrm{arctan}(2\eps^{1/2})\leq 3\eps^{1/2}$ from the Taylor expansion of $\tan(x)$ near $0$.
By Lemma~\ref{lem:angle2} the stretch factor of the $as_V$-path and the path $(a,s_V,b)$ are each at most $1+O(\eps)$. Lemma~\ref{lem:combine6} provides a path from $s_V$ to $b$ with stretch factor $1+O(\eps)$. Overall, $\|P_{ab}\|\leq (1+O(\eps))\|ab\|$.
\end{proof}

The combination of Lemmas~\ref{lem:dir2} and~\ref{lem:tame-staircase} provides a directional $(1+\eps)$-spanner for
all point pairs on the boundary of a tame histogram.

\begin{corollary}\label{cor:dir3}
Let $H$ be a tame histogram and  $S\subset \partial H$ a finite point set.
Then there exists a geometric graph $G$ of weight $\|G\|=O(\eps^{-1/2}\,\hper(H))$
such that $G$ contains a $ab$-path $P_{ab}$ with $\|P_{ab}\|\leq (1+O(\eps))\|ab\|$
for all $a,b\in S$ whenever $ab\subset H$ and $|\slope(ab)|\geq \eps^{-1/2}$.
\end{corollary}

\section{Directional Spanners for Thin Histograms}
\label{sec:thin}

We can now construct a directional $(1+\eps)$-spanner for a thin histogram.

\begin{lemma}\label{lem:dir1}
Let $F$ be a thin histogram and $S\subset \partial F$ a finite point set.
Then there exists a geometric graph $G$ of weight $\|G\|=O(\eps^{-1/2}\,\hper(F))$
such that $G$ contains a $ab$-path $P_{ab}$ with $\|P_{ab}\|\leq (1+\eps)\|ab\|$
for all $a,b\in S$ if $ab\subset F$ and $\mathrm{dir}(ab)\in D$.
\end{lemma}
\begin{proof}
Let $F$ be a thin histogram bounded by a vertical segment $pq$ and a $y$-monotone $pq$-path. By the definition of thin histograms, for all vertical chords $ab$ with $a,b\in L\cap S$, we have $\|L_{ab}\|\leq \frac12\eps^{1/2}\, \|ab\|$.
We construct directional spanner in two steps.

\smallskip\noindent\textbf{Case~1: Directional spanners for chords $ab$, with $a,b\in L$.}
Similarly to the proof of Lemma~\ref{lem:monotone}, let $L'$ be the unfolding of $L$ into a staircase path; refer to Figs.~\ref{fig:thin}(a)--(b). For every point $p\in L$, let $p'$ denote the corresponding point in $L'$. Denoting by $R$ and $R'$ the vertical strips spanned by $L$ and $L'$, respectively, there is a continuous and piecewise isometric function $\varrho:R'\rightarrow R$ such that $\varrho(L')=L$.
Any chord $ab$ of $F$ with $a,b,\in L$ corresponds to a segment $a'b'$ with $a',b' \in L'$, where $L'_{a'b'}$ denotes the subpath of $L'$ between $a'$ and $b'$.

\begin{figure}[htbp]
 \centering
 \includegraphics[width=0.95\textwidth]{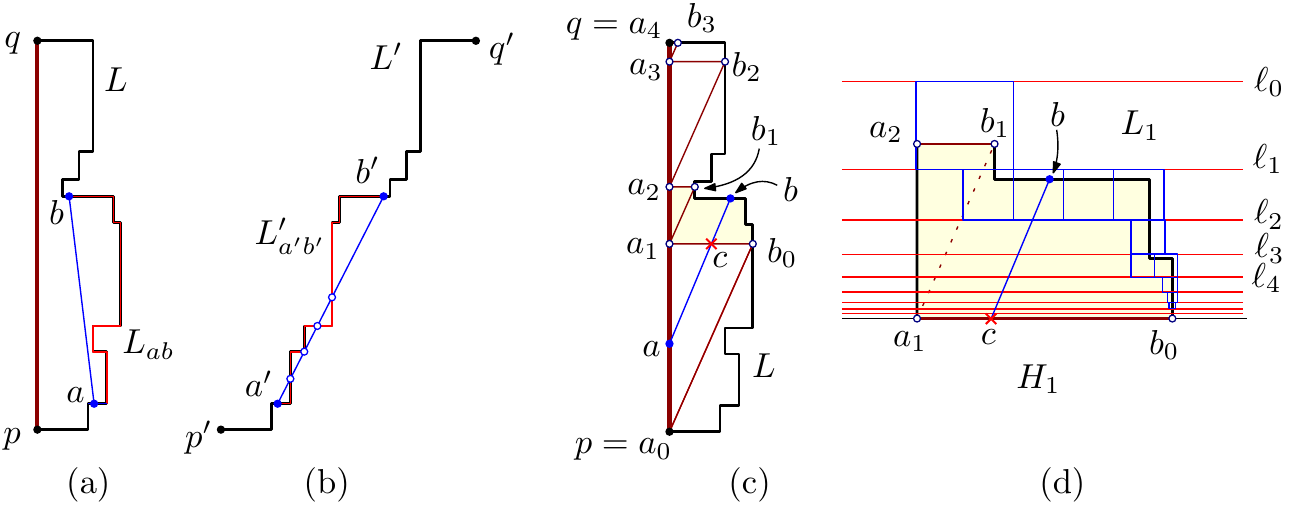}
 \caption{(a) A thin histogram, and a chord $ab$ with $|\slope(ab)|\geq \eps^{-1/2}$.
 (b) The $y$-monotone path $L$ is unfolded into a staircase path $L'$.
 (c) A $pq$ path $(a_0,b_0,a_1,\ldots, a_4)$.
 (d) A covering of $L_i$ with squares.}
    \label{fig:thin}
\end{figure}

Let $ab$ be a chord of $F$ with $a,b\in L$ and $|\slope(ab)|\geq \eps^{-1/2}$.
Then $\hght(ab)\geq \eps^{-1/2}\wdth(ab)$. Since $F$ is a thin histogram, then
\begin{align*}
\hper(L_{ab})
&\leq \wdth(ab)+\frac{\eps^{1/2}}{2}\, \hght(ab)\\
&\leq \left(\frac{1}{|\slope(ab)|} + \frac{\eps^{1/2}}{2}\right)\hght(ab)\\
&\leq 2\,\eps^{1/2}\, \vper(L_{ab}).
\end{align*}
Consequently,
\begin{equation}\label{eq:slope'}
    |\slope(a'b')|
    \geq \frac{\vper(L_{a'b'})}{\hper(L_{a'b'})}
    = \frac{\vper(L_{ab})}{\hper(L_{ab})}
    \geq \frac{\eps^{-1/2}}{2} .
\end{equation}
This in turn implies
\begin{align*}
\|a'b'\|
    & = \Big( (\wdth(a'b'))^2+ (\hght(a'b'))^2 \Big)^{1/2}
      = \left(1+\frac{1}{(\slope(a'b'))^2} \right)^{1/2} \hght(a'b') \\
    & \leq \left(1+4\eps\right)^{1/2} \hght(ab)
      \leq \left(1+O(\eps)\right) \|ab\|.
\end{align*}

We are now ready to construct a spanner, using Lemma~\ref{lem:staircase}.
Let $S'=\{s'\in L': s\in S\cap L\}$ be the set of points in $L'$ corresponding to the points in $S\cap L$. Let $S''$ be union of $S'$ and all intersection points between $L'$ and the line segments spanned by $S'$. Applying Lemma~\ref{lem:staircase} for $L'$ and the point set $S''$ twice (both above and below $L'$),  we obtain a geometric graph $G'$ of weight $O(\eps^{-1/2}\,\wdth(L'))$ that contains, for all chords $p'q'$ of $L'$ with $p',q'\in S''$ and $|\slope(p'q')|\geq 2\,\eps^{-1/2}$, a path $P_{p'q'}$ of weight at most $\|P_{p'q'}\|\leq (1+\eps)\|p'q'\|$.

Given a point pair 

\smallskip\noindent\emph{Lightness analysis in Case~1.}
Let $G=\varrho(G')$, which has the same weight as $G'$, that is,
    \[
    \|G\|=\|G'\|
    =O(\eps^{-1/2}\,\wdth(L'))
    = O(\eps^{-1/2}\,\hper(L))
    = O(\eps^{-1/2}\,\hper(F)).
    \]

\smallskip\noindent\emph{Stretch analysis in Case~1.}
Let $ab$ be a chord of $L$ with $|\slope(ab)|\geq \eps^{-1/2}$. Then  $|\slope(a'b')|\geq \frac12\,\eps^{-1/2}$ by Equation~\eqref{eq:slope'}. The line segment $a'b'$ is not necessarily a chord of $L'$. The staircase path $L'$ subdivides $a'b'$ into a chain $(p'_0,\ldots ,p'_m)$ of collinear chords of $L'$. For all $i=1,\ldots ,m$, the spanner $G'$ contains a $p'_{i-1}p'_i$-path $P'_i$ of weight $(1+O(\eps))\|p'_{i-1} p'_i\|$. The concatenation of these paths is an
$a'b'$-path $P'_{a'b'}$ of weight $(1+O(\eps))\|a'b'\|$. Finally, $G=\varrho(G')$
contains the $ab$-path $P_{ab}=\varrho(P'_{a'b'})$ of weight
    \[
    \|P_{ab}\|=\|P_{a'b'}\|
    \leq (1+O(\eps))\|a'b'\|
    \leq (1+O(\eps))^2\|ab\|
    \leq (1+O(\eps))\|ab\|,
    \]
as required.

\smallskip\noindent\textbf{Case~2: Directional spanner for chords between $pq$ and $L$.}
Assume w.l.o.g.\ that $pq$ is the left edge of $F$, and $p$ is the bottom vertex of $pq$.
We describe a construction for chords $ab$ with $\slope(ab)\geq \eps^{-1/2}$;
the construction is analogous for $\slope(ab)\leq -\eps^{-1/2}$, after a reflection.

We subdivide $F$ by a $pq$-path constructed recursively as follows;
see Fig.~\ref{fig:thin}(c).
Initially, we set $i=0$ and $a_0=p$. While $a_i\neq q$, we construct
point $b_i\in L$ such that $\slope(a_ib_i)=\eps^{-1/2}$;
and then construct $a_{i+1}\in pq$ such that $b_ia_{i+1}$ is horizontal.
Since $P$ has finitely many vertices, the algorithm terminates
with $r=a_q$ for some integer $m\geq 1$.
For short, denote by $L_i$ the subpath of $L$ between $b_{i-1}$ and $b_i$,
and let $H_i$ be the $y$-monotone histogram bounded by $L_i$ and the
path $(b_{i-1},a_i,a_{i+1},b_i)$.

For every $i=0,\ldots , m-1$, we construct a geometric graph $G_i$ as follows.
Graph $G_i$ includes the boundary of $H_i$, that is, $L_i$ and the path $(b_{i-1},a_i,a_{i+1},b_i)$.
The segments $a_{i-1}a_i\cup b_{i-1}a_i$ form a staircase path; by Lemma~\ref{lem:staircase}, there exists a geometric graph of weight $O(\|a_ia_{i+1}\|)$ that is a directional $(1+\eps)$-spanner for chords of $a_{i-1}a_i\cup b_{i-1}a_i$ of slope $\eps^{-1/2}$ or more; we add this graph to $G_i$.
Finally, similarly to the proof of Lemma~\ref{lem:dir2},
we cover $H_i$ with squares.
Specifically, for every $j\in \mathbb{N}$, let
\[\ell_j:y=y(a_i)+\hght(H_i) \left(1-3\cdot \sqrt{\eps}\right)^i.\]
We tile the horizontal strip between two consecutive lines, $\ell_j$ and $\ell_{j+1}$, by squares see Fig.~\ref{fig:thin}(d). For each square $Q$ that intersects $L_i$, Lemma~\ref{lem:combine5} (invoked with $36\,\eps$ in place of $\eps$) yields a geometric graph of weight $O(\eps^{-1/2}\, \wdth(Q)) = O(\eps^{-1/2}\, \hper(L\cap Q))$
that contains, for every chord $st$ of $H_i$ with $s\in b_{i-1}a_i$ and $t\in L_i\cap Q$,
an $st$ path of weight $(1+O(\eps))\|st\|$. We add all these graphs to $G_i$.

\smallskip\noindent\emph{Lightness analysis in Case~2.}
We have subdivided the edge $pq$ into a path $(a_0,\ldots ,a_m)$,
hence $\sum_{i=1}^m \|a_{i-1}a_i\|= \|pq\|$. Consequently,
$\sum_{i=1}^m \|b_{i-1}a_i\| =\sum_{i=1}^m \eps^{1/2} \, \|a_{i-1}a_i\| = \eps^{1/2}\|pq\|$. The total weight of the spanners between $a_{i-1}a_i$ and $b_{i-1}a_i$, for $i=1,\ldots, m$, is also bounded by $\sum_{i=1}^m O(\|a_{i-1}a_i\|) = O(\|pq\|)$.
Finally, the path $L$ is covered by squares, and for each
square $Q$ with $Q\cap L\neq \emptyset$, we have added a graph of weight
$O(\eps^{-1/2}\, \hper(L\cap Q))$. Summation over all squares yields
    \[
    \sum_Q O(\eps^{-1/2}\, \hper(L\cap Q))
    = O(\eps^{-1/2}\, \hper(L))
    = O(\eps^{-1/2}\, \hper(F)).
    \]
As $\|pq\|=O(\eps^{-1/2}\, \hper(F))$ in a thin histogram, then the total weight of the spanner for $F$ is $O(\eps^{-1/2}\, \hper(F))$.

\smallskip\noindent\emph{Stretch analysis in Case~2.}
Let $a\in pq$ and $b\in L$ such that $|\slope(ab)|\geq \eps^{-1/2}$.
By symmetry, we may assume $\slope(ab)\geq \eps^{-1/2}$.
Since $\slope(a_ib_i)=\eps^{-1/2}$ for all $i=0,\ldots , m-1$,
then $ab$ cannot cross any of these segments, and so $ab$
crosses at least one segment $a_ib_{i-1}$.
Let $j$ be the largest index such that $ab$ crosses $a_jb_{j-1}$,
and let $c=ab\cap a_jb_{j-1}$; see Fig.~\ref{fig:thin}(c).
If $a\in a_{j-1}a_j$, then we find a path $P_{ab}$ as
a concatenation of an $ac$-path $P_{ac}$ using the
spanner for the staircase path $(a_{j-1},a_j,b_{j-1})$,
and a $cb$-path $P_{cb}$ using the spanner in the histogram $H_j$.
Then
    \[
    \|P_{ab}\| =\|P_{ac}\|+\|P_{cb}\|
    \leq (1+O(\eps))(\|ac\| + \|cb\|)=(1+O(\eps))\|ab\|.
    \]
If $a$ is below point $a_{j-1}$, then we construct an $ab$-path $P_{ab}$
as a concatenation of edge $aa_{j-1}$, followed by a path from
$a_{j-1}$ to $b$ via $c$ as in the previous case.
The slope of every edge of the path $(a,a_{j-1},c,b)$ is more than
$\eps^{-1/2}$, hence $\|aa_{j-1}\|+\|a_{j-1}c\|+\|cb\|\leq (1+\eps)\|ab\|$
by Lemma~\ref{lem:angle2}. The spanner contains paths that
approximate $a_{j-1}c$ and $cb$ by a factors of $1+O(\eps)$. Overall, we have
$\|P_{ab}\|\leq (1+O(\eps))\|ab\|$, as required.
\end{proof}

Corollary~\ref{cor:dir3} and Lemma~\ref{lem:dir1} jointly imply Lemma~\ref{lem:hist}.

\histlemma*

This completes all components needed for Theorem~\ref{thm:UB}.

\section{Conclusion and Outlook}
\label{sec:cons}

We have studied Euclidean Steiner $(1+\eps)$-spanners under two optimization
criteria, \emph{lightness} and \emph{sparsity}, and obtained improved lower and upper bounds. In Euclidean $d$-space, the same point sets (grids in two parallel hyperplanes) establish the lower bounds $\Omega(\eps^{-d/2})$ for lightness and $\Omega(\eps^{(1-d)/2})$ for sparsity, for Steiner $(1+\eps)$-spanners
(cf.\ Theorem~\ref{thm:lb}).
For lightness, we obtained a matching lower bound of $O(\varepsilon^{-1})$ in the plane (cf.\ Theorem~\ref{thm:upper}). However, in dimensions $d \ge 3$, a $\tilde{\Theta}(\varepsilon^{-1/2})$-factor gap remains between the current upper bound $\tilde{O}(\varepsilon^{-(d+1)/2})$ \cite[Theorem~1.7]{le2020unified} and the lower bound $\Theta(\varepsilon^{-d/2})$ of Theorem~\ref{thm:lb}.
Le and Solomon~\cite[Theorem~1.3]{le2019truly} constructed spanners with sparsity $\tilde{O}(\varepsilon^{(1-d)/2})$, matching the lower bound up to lower-order factors in every dimension $d\in \mathbb{N}$.

Without Steiner points, the greedy-spanner achieves the worst-case lower bounds of $\Omega(\eps^{-d})$ and $\Omega(\eps^{-d+1})$ for lightness and sparsity, resp., in every dimension $d\geq 2$. When Steiner points are allowed, however, it is unclear whether a  $(1+\eps)$-spanner can meet both optimization criteria. The current best constructions for sparsity  (\cite[Theorem~1.3]{le2019truly}) and lightness (Theorem~\ref{thm:upper} and \cite[Theorem~1.7]{le2020unified}) place Steiner points in $d$-space to optimize one criterion, but not the other.
We conjecture that a Euclidean Steiner $(1+\varepsilon)$-spanner cannot simultaneously attain both lower bounds (lightness and sparsity) of Theorem~\ref{thm:lb}. 
Exploring the trade-offs between lightness and sparsity in Euclidean $d$-space remains an open problem.

In the plane, in particular, we have proved a tight upper bound of $O(\eps^{-1})$ on the lightness of Euclidean Steiner $(1+\eps)$-spanners (cf. Theorem~\ref{thm:upper}). Our proof is constructive: For every finite set $S\subset \mathbb{R}^2$, we describe a Euclidean Steiner $(1+\eps)$-spanner of weight $O(\eps^{-1}\,\|\MST(S)\|)$. However, we do not control the number of Steiner points. This immediately raises two questions: What is the minimum number of Steiner points and what is the minimum sparsity of a Euclidean Steiner $(1+\eps)$-spanner of lightness $O(\eps^{-1})$ that can be attained for all finite point sets in the plane?

Planarity is an important aspect of any geometric network in $\mathbb{R}^2$.
It is desirable to construct Euclidean $(1+\eps)$-spanners that are plane, i.e., no two edges of the spanner cross. Any Steiner spanner can be turned into a plane spanner (planarized), with the same weight and the same spanning ratio between the input points, by introducing Steiner points at all edge crossings. However, planarization may substantially increase the number of Steiner points. Bose and Smid~\cite[Sec.~4]{BoseS13} note that Arikati et al.~\cite{ArikatiCCDSZ96} constructed a Euclidean plane $(1+\eps)$-spanner with $O(\eps^{-4} n)$ Steiner points for $n$ points in $\mathbb{R}^2$; see also~\cite{MaheshwariSZ08}. Borradaile and Eppstein~\cite{BorradaileE15} improved the bound to $O(\eps^{-3} n\log \eps^{-1})$ in certain special cases where all Delaunay faces of the point set are fat. It remains an open problem to find the optimum dependence of $\eps$ for plane Steiner $(1+\eps)$-spanners; and for plane Steiner $(1+\eps)$-spanners of lightness $O(\eps^{-1})$.

\paragraph*{Acknowledgements.} We thank the anonymous reviewers of earlier versions of this paper for many helpful comments and suggestions that helped clarify the presentation. 

\bibliographystyle{plainurl}
\bibliography{main}

\end{document}